\documentclass[draft]{IEEEtran}
\ifCLASSINFOpdf
\else
\fi
\usepackage[cmex10]{amsmath}
\usepackage{amsfonts}
\usepackage{amssymb}
\usepackage[top=1in, bottom=1in, left=1.25in, right=1.25in]{geometry}
\usepackage[final,hyperfootnotes=false]{hyperref}
\usepackage[final]{graphicx}
\usepackage{tikz}
\usepackage{pgfplots}
\usetikzlibrary{shapes,arrows,positioning,fit,backgrounds}
\usetikzlibrary{calc}

\newtheorem{theorem}{Theorem}
\newtheorem{lemma}{Lemma}
\newtheorem{definition}{Definition}
\newtheorem{corollary}{Corollary}

\newtheorem{proposition}{Proposition}
\newtheorem{remark}{Remark}

\newtheorem{counter-example}{Counter-example}
\onecolumn
\DeclareMathOperator*{\supp}{supp}
\DeclareMathOperator*{\interior}{int}
\DeclareMathOperator*{\spn}{span}
\DeclareMathOperator*{\sign}{sign}
\DeclareMathOperator*{\dist}{dist}

\newcommand{\mb}{\mathbf}

\begin{document}
\title{Robustness of Sparse Recovery via $F$-minimization: A Topological Viewpoint}
\author{Jingbo~Liu,~\IEEEmembership{Student Member,~IEEE,}
        Jian~Jin,~
        and~Yuantao~Gu,~\IEEEmembership{Member,~IEEE}
\thanks{J. Liu is now with the Department of Electrical Engineering, Princeton University, Princeton, NJ 08544 USA (e-mail: jingbo@princeton.edu). Main part of this work was done when he was with the Department
of Electronic Engineering, Tsinghua University, Beijing, 100084 China. J. Jin and Y. Gu are both with the Department of Electronic Engineering, Tsinghua University (e-mail: jinjian620@gmail.com,~gyt@tsinghua.edu.cn). The corresponding author of this paper is Yuantao Gu.

This paper was presented in part at IEEE International Symposium on
Information Theory (ISIT) in Istanbul, Turkey, 2013.}
}
\markboth{
}%
{Shell \MakeLowercase{\textit{et al.}}: Robustness of Sparse Recovery via $F$-minimization: A Topological Viewpoint}
%



\maketitle

\begin{abstract}
A recent trend in compressed sensing is to consider non-convex optimization techniques for sparse recovery. The important case of $F$-minimization has become of particular interest, for which the exact reconstruction condition (ERC) in the noiseless setting can be precisely characterized by the null space property (NSP). However, little work has been done concerning its robust reconstruction condition (RRC) in the noisy setting. We look at the null space of the measurement matrix as a point on the Grassmann manifold, and then study the relation between the ERC and RRC sets, denoted as $\Omega_J$ and $\Omega_J^r$, respectively. It is shown that $\Omega_J^r$ is the interior of $\Omega_J$, from which a previous result of the equivalence of ERC and RRC for $\ell_p$-minimization follows easily as a special case. Moreover, when $F$ is non-decreasing, it is shown that $\overline{\Omega}_J\setminus\interior(\Omega_J)$ is a set of measure zero and of the first category. As a consequence, the probabilities of ERC and RRC are the same if the measurement matrix $\mathbf{A}$ is randomly generated according to a continuous distribution. Quantitatively, if the null space $\mathcal{N}(\bf A)$ lies in the ``$d$-interior'' of $\Omega_J$, then RRC will be satisfied with the robustness constant $C=\frac{2+2d}{d\sigma_{\min}(\mathbf{A}^{\top})}$; and conversely if RRC holds with $C=\frac{2-2d}{d\sigma_{\max}(\mathbf{A}^{\top})}$, then $\mathcal{N}(\bf A)$ must lie in $d$-interior of $\Omega_J$. We also present several rules for comparing the performances of different cost functions. Finally, these results are capitalized to derive achievable tradeoffs between the measurement rate and robustness with the aid of Gordon's escape through the mesh theorem or a connection between NSP and the restricted eigenvalue condition.
\end{abstract}

\begin{IEEEkeywords}
Reconstruction algorithms,
compressed sensing,
minimization methods,
robustness,
null space
\end{IEEEkeywords}

%
\IEEEpeerreviewmaketitle

\section{Introduction}
\IEEEPARstart{C}{ompressed} Sensing is a method of recovering a sparse signal from a set of under-determined linear measurements. Ideally, the sparsest solution is given by the $\ell_0$-norm minimization method:
\begin{equation}
\min_{\mathbf{x}\in \mathbb{R}^n}~\|\mathbf{x}\|_0~\textrm{s.t.}~\bf y=Ax,
\end{equation}
where $\mathbf{A}$ is an $m\times n$ measurement matrix, $\mathbf{y}\in \mathbb{R}^m$ is the linear measurements, and we assume that $m<n$. It is well known that exactly solving the $\ell_0$-minimization is computational intractable since it is a hard combinatorial problem \cite{DC}. Therefore, many algorithms have been proposed to reduce the computational complexity. Roughly speaking, these algorithms fall into two categories: 1) minimization techniques, where the sparse solution is retrieved by minimizing an appropriate cost function \cite{Donoho,zap}, and 2) greedy pursuits, a representative of which is the orthogonal matching pursuit (OMP) \cite{OMP}.

In general, the greedy algorithms often incur less computational complexity, but the minimization techniques are more advantageous in terms of accuracy. The most basic minimization technique is the $\ell_1$-minimization, also known as Basis Pursuit (BP) \cite{DC,Donoho,donoho2}:
\begin{equation}\label{min1}
\min_{\mathbf{x}\in \mathbb{R}^n}\|\mathbf{x}\|_1\quad \textrm{s.t.}~\mathbf{y}=\mathbf{A}\mathbf{x},
\end{equation}
which is a simple convex optimization and can be recast as a linear program. Recently there is a trend to consider minimizing non-convex cost functions. Examples include:

$\bullet$ $\ell_p$ cost function. The $\ell_p$-minimization ($0<p<1$) \cite{chartrand,Foucart1,Gribonval} considers an optimization problem similar to (\ref{min1}) but the cost function is replaced with $\|\mathbf{x}\|^p_p$.

$\bullet$ Approximate $\ell_0$ cost functions, such as those in the zero point attracting projection (ZAP), \cite{zap}, and smooth $\ell_0$ algorithm \cite{sl0}. Also for statisticians, smoothly clipped absolute deviation (SCAD) penalty \cite{Fan} and the minimax concave penalty (MCP) \cite{MCP} are familiar concave penalties used for variable selection.

Although the non-convex nature of these cost functions makes it difficult to exactly solve the corresponding optimization problems, various practical algorithms can be adapted to these non-convex problems, including the iteratively re-weighted least squares minimization (IRLS) \cite{chartrand1,irls}, iterative thresholding algorithm (IT) \cite{it}, which are based on fixed point iteration; and the zero point attracting projection algorithm (ZAP) \cite{zap,wang,laming}, which is based on Newton's method for solving nonlinear optimization. In general the non-convex algorithms have empirically outperformed BP in the various respects, because nonlinear cost functions can better promote sparsity than the $\ell_1$ cost function. Thus, a detailed study of the reconstruction properties of these sparse recovery methods remains important.

Most of these non-convex optimizations can be subsumed in a general category called ``$F$-minimization'' \cite{lqnsp}, in which the cost function satisfies some desirable properties, such as subadditivity. The precise definition of the class of cost functions of our interest will be given in the next section.

Two concepts arise naturally in the compressed sensing problem: The \emph{exact recovery condition} (ERC) in the noiseless setting and the \emph{robust recovery condition} (RRC) in the noisy setting. In the literature, ERC typically requires that all sparse signals can be exactly recovered. In addition to this, RRC requires that if the measurement is noisy, the reconstruction error is bounded by the norm of the noise vector multiplied by a constant factor.

While the rigorous definitions of ERC and RRC are deferred to Section \ref{sec2}, we remark here in passing that RRC trivially implies ERC, because ERC can be seen as a special case of RRC where the measurement is free of noise.
Conversely, it is not obvious whether ERC also implies RRC, or RRC is \emph{strictly} stronger than ERC. Early work in compressed sensing have provided sufficient conditions for ERC and RRC of the $\ell_1$-minimization, based on the so-called restricted isometry property (RIP) \cite{DC}, and those sufficient conditions appear to be identical.
However, analysis based on RIP generally fails to provide exact (necessary and sufficient) condition for ERC and RRC. Another line of research has considered the null space property (NSP), which gives a both necessary and sufficient condition for ERC of the $\ell_p$-minimization. On the other hand, the connection between NSP and RRC is generally much less known: ``While the NSP is both necessary and sufficient for establishing guarantees of..., these guarantees do not account for \emph{noise}'' \cite[Section~1.4.2]{davenport2011introduction}.
However \cite{Foucart} provided a sufficient condition, called NSP', for RRC of $\ell_p$-minimization. Later Aldroubi et al.~proved in \cite{lqnsp} that NSP and NSP' are in fact equivalent. Hence, we have that ERC and RRC are actually the same condition for $\ell_p$-minimization.

In contrast to the special case of $\ell_p$-minimization, the robust recovery condition for the more general case of $F$-minimization has been recognized as ``not easy to establish'' \cite{lqnsp}, merely based on the idea of NSP. The fundamental issue of robustness in $F$-minimization has remained relatively unexplored.

The primary purpose of this paper is to give an exact characterization of the relationship between ERC and RRC in the general $F$-minimization problem. We first show that ERC and RRC depends only on the configuration of the null space of the measurement matrix (the entire entries of the matrix is of course sufficient, but not necessary, information).
Moreover, since the null spaces are linear subspaces of the Euclidean space, they can be viewed as points on a Grassmann manifold, which has a natural topological structure,
hence concepts such as open sets and interior are well defined for collections/sets of the null spaces. We denote by $\Omega_J$ and $\Omega_J^{r}$ the sets that consist of the null spaces satisfying ERC and RRC for the $F$-minimization, respectively.
We show that $\Omega_J^{r}$ is exactly the interior of $\Omega_J$ (Theorem \ref{th2}).
Hence we can give an alternative proof of the equivalence of ERC and RRC in $\ell_p$-minimization,
by simply showing that $\Omega_J$ is open in this special case.
We would like to remark that this analytical framework also gives rise to new ideas and results, including:

\begin{figure}
\begin{center}
\begin{tikzpicture}
\node[rectangle] (rrc0) {$C=\frac{2-2d}{d\sigma_{\max}(\mathbf{A}^{\top})}$};
\node[rectangle] (int) [above=-0.05cm of rrc0] {$d\mbox{-}\interior(\Omega_J)$};
\node[rectangle] (rrc1) [above=0cm of int] {$C=\frac{2+2d}{d\sigma_{\min}(\mathbf{A}^{\top})}$};
\node[rectangle] (rrc) [above=0cm of rrc1] {RRC};
\node[rectangle] (erc) [above=0cm of rrc] {ERC};

\node[draw=black,inner sep=-4pt,thick,ellipse,fit=(rrc0)] (nrrc0) {};
\node[draw=black,inner sep=-4pt,thick,ellipse,fit=(nrrc0) (int)] (nint) {};
\node[draw=black,inner sep=-7pt,thick,ellipse,fit=(nint) (rrc1)] (nrrc1) {};
\node[draw=black,inner sep=-10pt,thick,ellipse,fit=(nrrc1) (rrc)] (nrrc) {};
\node[draw=black,inner sep=-13pt,thick,ellipse,fit=(nrrc) (erc)] (nerc) {};
\end{tikzpicture}
\caption{Relationship between subsets of $G_l(\mathbb{R}^n)$.}
\end{center}
\label{fig0}
\end{figure}
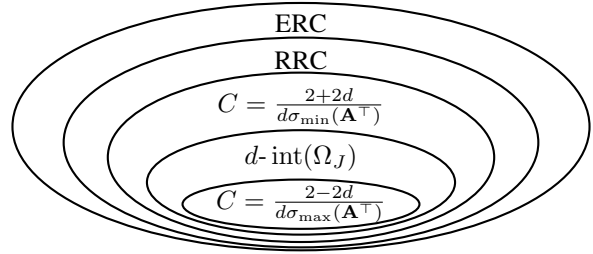

\begin{enumerate}
  \item Equivalence of ERC and RRC in probability. Under some mild assumptions we show that $\Omega_J$ and $\Omega_J^{r}$ are ``almost equal'' in the sense that the difference set is of measure zero and of the first category. Building on this, we show that ERC and RRC hold true with the same probability if the measurement matrix is randomly generated according to a continuous distribution.
  \item Comparison between different sparseness measures. It is interesting and valuable to know how the performances between different sparseness measures compare. Gribonval et al.~\cite[Lemma 7]{Gribonval} provided a condition under which one spareness measure is better than another in the sense of ERC. Combining this with our result, we show that this condition also provides a comparison in terms of RRC. Moreover, with the concept of measure zero set on the Grassmannian, we are able to provide additional comparison rules which guarantee that one sparse measure is better than the other in terms of probability of ERC/RRC.
  \item Tradeoff between measurement rate and robustness. We show that a matrix whose null space falls in the ``$d$-interior'' of $\Omega_J$ satisfies RRC with the robustness constant $C=\frac{2(1+d)}{d\sigma_{\min}(\mathbf{A}^{\top})}$. Conversely, $C=\frac{2(1-d)}{d\sigma_{\max}(\mathbf{A}^{\top})}$ implies that the null space lies in the ``$d$-interior''; see Figure~\ref{fig0}. This result can be seen as a quantitative version of the aforemention interior characterization of the RRC set, and can be combined with 2) to compute achievable tradeoffs between the measurement rate and robustness: for rotationally invariant matrix ensembles, Gordon's escape through the mesh theorem can be used to upper bound the measure of the $d$-interior. To illustrate this method, we derive the tradeoff when $F(x)/x$ is non-increasing in the asymptotic linear growth case. For matrices satisfying the restricted isometry property (RIP), a tradeoff can also be derived using a connection between RIP, the restricted eigenvalue condition and NSP.
\end{enumerate}

The rotationally invariance in 3) means that the distribution of the matrix is invariant under right multiplication with any orthogonal matrix. This is an nice property, not only because it is satisfied by important matrix ensembles such as the standard Gaussian ensemble, but also because it can be employed to construct universal encoders \cite{tao1}: suppose a signal $\mb{s}$ is sparse in a certain basis $\bf\Phi$, i.e.\
\begin{align}
\bf s=\Phi x,
\end{align}
where $\bf{x}$ is a sparse vector and $\bf\Phi$ is an orthogonal matrix known to the decoder but not the encoder. If $\bf{A}$ satisfies rotational invariance, the encoder can always take the measurement
\begin{align}
\bf As=A\Phi x,
\end{align}
and the unavailability of the side information $\bf\Phi$ to the encoder does not matter since $\bf A\Phi$ has the same distribution as $\bf A$.

The organization of the paper is as follows. In Section \ref{sec2} we present the mathematical formulation of the problem and a brief introduction to null space property and the Grassmann manifold. Section \ref{secrelation} studies the relationship between ERC and RRC: Section~\ref{pa} gives an exact characterization of RRC set as the interior of ERC set on the Grassmannian; in Section~\ref{s3b} we show than the ERC and RRC sets differ by a set of measure zero and of the first category; Section~\ref{s3c} provides quantitative results of the robustness of the measurement matrix whose null space lies in $d\mbox{-}\interior(\Omega_J)$ (the $d$-interior mentioned earlier). In Section~\ref{secrules} we provide some rules for comparing the performance of different sparse measures. Utilizing results from \ref{s3c} and \ref{secrules}, Section~\ref{sgordon} and Section~\ref{sbeyond} provide two approaches of estimating the probability of $d\mbox{-}\interior(\Omega_J)$ and deriving the tradeoffs between the measurement rate and the robustness. Section~\ref{comp} compares our approaches and definitions with related ones in the literature. Finally in Section~\ref{conclusion} we conclude by reviewing the results and pointing out possible directions for future work.

\section{Problem Setup and Key Definitions}\label{sec2}
This section provides the mathematical formulation of the problem and the definitions of some key concepts. We shall use lower case bold letters for vectors, and upper case bold letters for matrices. Notation $\mathbb{M}(m,n)$ denotes the set of $m\times n$ real matrices. Throughout the paper we suppose the observation matrix is $m\times n$, and set $l:=n-m$, unless otherwise indicated. $\|\mathbf{x}\|_0$ refers to the $\ell_0$ norm\footnote{Strictly speaking, the $\ell_0$ norm and $\ell_p(0<p<1)$ norm defined here do not satisfy the definition of norm in mathematics.} of $\mathbf{x}$, i.e., the number of non-zero elements in the vector, and $\|\mathbf{x}\|_p:=(\sum_k|x(k)|^p)^{1/p}$ denotes the $\ell_p$ norm of $\mathbf{x}$.
\subsection{Basic Model}
Let $\mathbf{\bar{x}}\in \mathbb{R}^n$, $\mathbf{A}\in \mathbb{M}(m,n)$, $\mathbf{v}\in \mathbb{R}^m$ be the sparse signal, the measurement matrix, and the additive noise, respectively. Let $T:=\supp(\mathbf{\bar{x}})$ be the support of $\mathbf{\bar{x}}$. Vector $\mathbf{\bar{x}}$ is called $k$-sparse if $|T|\le k$. The linear measurement $\mathbf{y}$ is given by
\begin{equation}
\mathbf{y}=\mathbf{A}\mathbf{\bar{x}}+\mathbf{v}.
\end{equation}

We consider the problem of recovering $\mathbf{\bar{x}}$ through an optimization. Supposing $F:[0,+\infty)\to[0,+\infty)$ is a given function, we define the cost function
\begin{equation}\label{cost}
J({\mathbf{x}}):=\sum_{k=1}^n F(|x(k)|).
\end{equation}
With a slight abuse of the notation, we shall also use the notations:
\begin{align}
J(\mathbf{x}_T):&=\sum_{k\in T} F(|x(k)|),\nonumber\\
J({\mathbf{x}}_{T^c}):&=\sum_{k\in T^c} F(|x(k)|),\nonumber
\end{align}
where $\mathbf{x}_T\in\mathbb{R}^{|T|},~
\mathbf{x}_{T^c}\in\mathbb{R}^{n-|T|}$ denote the restriction of $\mathbf{x}$ on the set $T,~T^c$, respectively. Clearly (\ref{cost}) is a very general model: For example, if one chooses $F(x)=1_{x>0}$ then $J(\mathbf{x})=\|\mathbf{x}\|_0$; if $F(x)=x^p$ then $J(\mathbf{x})=\|\mathbf{x}\|^p_p$.

The conditions ERC and RRC are commonly formulated as follows, see for example \cite{Donoho}\cite{tao1}\cite{lqnsp}.
\begin{definition}[Exact recovery condition]
In the noiseless case, the sparse signal is retrieved via the following optimization:
\begin{equation}\label{minimization}
\min_{\mathbf{x}\in\mathbb{R}^n}J({\mathbf{x}})\quad \textrm{s.t.}~\mathbf{A}\mathbf{x}=\mathbf{y}.
\end{equation}
We say $\mathbf{A},~J$ satisfy the \emph{exact recovery condition} (ERC) if for any measurement $\mathbf{y}=\mathbf{A}\mathbf{\bar{x}}$, where $\mathbf{\bar{x}}$ is $k$-sparse, the vector $\mathbf{\bar{x}}$ is also the unique solution to (\ref{minimization}).
\end{definition}

\begin{definition}[Robust recovery condition]
In the noisy measurement ($\mathbf{v}\neq \mathbf{0}$) case, the sparse signal is retrieved via the following optimization:
\begin{equation}\label{minimization2}
\min_{\mathbf{x}\in\mathbb{R}^n}J({\mathbf{x}})\quad \textrm{s.t.}~\|\mathbf{A}\mathbf{x}-\mathbf{y}\|<\epsilon,
\end{equation}
where $\epsilon\in \mathbb{R}^+$ is a constant chosen to tolerate the noise. We say that the \emph{robust recovery condition} (RRC) is satisfied if the following holds. For any $k$-sparse signal $\mathbf{\bar{x}}$, noise $\mathbf{v}$ and $\epsilon$ satisfying $\|\mathbf{v}\|\le \epsilon$, and feasible solution $\mathbf{\hat{x}}$ satisfying $J(\mathbf{\hat{x}})\le J(\mathbf{\bar{x}})$, we have
\begin{equation}\label{defc}
\|\mathbf{\bar{x}}-\mathbf{\hat{x}}\|<C\epsilon,
\end{equation}
where $C$ is a constant.
\end{definition}

We end this subsection by remarking that ERC, RRC, and the constant $C$ in the definition of RRC all depend only on $\mathbf{A},k,J$.

\subsection{Null Space Property}
The null space property \cite{gribonval1,nsp,Gribonval} is useful for the analysis of a special class of cost functions, which we introduce as follows:
\begin{definition}[sparseness measure]\label{def1}
Function
\begin{equation}
F:[0,+\infty)\to[0,+\infty)
\end{equation}
is called a \emph{sparseness measure} if the following two conditions are satisfied:

$\bullet$ $F(|\cdot|)$ is subadditive on $\mathbb{R}$, i.e. $F(|x+y|)\le F(|x|)+F(|y|)$ for all $x,y,z\in \mathbb{R}$;

$\bullet$ $F(x) = 0$ if and only if $x=0$.\\
We denote by $\mathcal{M}$ the set of all sparseness measures.\footnote{For our purpose, the definition of sparseness measure in this paper does not need to require that $F(x)/x$ is non-increasing. A comparison with other definitions of the sparseness measure is given in Section~\ref{s5b}.}
\end{definition}

In this paper we assume that the function $F$ is a sparseness measure as in Definition \ref{def1}. This is a rather loose assumption, so that the key optimization problems in many of the sparse recovery algorithms can be subsumed in our framework, including $\ell_p$-minimization and ZAP algorithm. The definition is also quite natural, since it can be checked that $F$ is a sparseness measure if and only if its corresponding cost function $J$ induces a metric on $\mathbb{R}^n$ via $d(\mathbf{x},\mathbf{y})
:=J(\mathbf{x}-\mathbf{y})$.

When $F\in\mathcal{M}$, the \emph{null space property} (NSP) turns out to be equivalent with ERC:
\begin{lemma}[Null space property \cite{Gribonval}(Lemma 6)]\label{nspcond}
If $F\in\mathcal{M}$, then a necessary and sufficient condition for ERC is
\begin{equation}
J(\mathbf{z}_T)<J(\mathbf{z}_{T^c}),\quad\forall \mathbf{z}\in \mathcal{N}(\mathbf{A})\setminus\{\mathbf{0}\},~T:|T|\le k.
\end{equation}
where $\mathcal{N}(\mathbf{A})$ denotes the null space of $\mathbf{A}$.
\end{lemma}

It's useful to define the \emph{null space constant \cite{Gribonval}}, especially when one wants to study $\ell_p$-minimization or to compare it with $F$-minimization:
\begin{definition}[Null space constant, NSC]\label{de_nsc}
Suppose $F\in \mathcal{M},~q\in(0,1]$. Define the null space constant is defined as:
\begin{equation}\label{nsc}
\theta_{J}:=\sup_{\mathbf{z}\in\mathcal{N}(\mathbf{A})\setminus\{\mathbf{0}\}}
\max_{|T|\le k}\frac{J(\mathbf{z}_{T})}{J(\mathbf{z}_{T^c})}.
\end{equation}
In the same spirit, we denote by $\theta_{\ell_p}$ the null space constant associated with $\ell_p$ cost function.
\end{definition}

The null space constant is closely associated with NSP, and hence characterizes the performance of $F$-minimization. We have the following result, which is a direct consequence of Definition \ref{de_nsc} and Lemma \ref{nspcond}.
\begin{lemma}\label{nspcond2}~\\
1) $\theta_J\le 1$ is a necessary condition for ERC;\\
2) $\theta_J <1$ is a sufficient condition for ERC.
\end{lemma}
In the case of $\ell_p$-minimization, one can obtain the following characterization (c.f.\cite{Foucart}), which is more exact than the case of $F$-minimization as described in Lemma \ref{nspcond2}:
\begin{lemma}\label{le1}
For $\ell_p$ cost functions, $\theta_{\ell_p}< 1$ is a both necessary and sufficient condition for ERC.
\end{lemma}

\subsection{Preliminaries of the Grassmann Manifold}\label{pc}
In this subsection we briefly review some relevant properties of the Grassmann manifold. More detailed treatment of this subject and related concepts in differential topology can be found in many standard texts, such as \cite{boothby,milnor}. The main thrust for considering this object is that, by Lemma~\ref{nspcond} the property of exact recovery of a particular measurement matrix is completely determined by its null space $\mathcal{N}(\mathbf{A})$, which is an $l:=n-m$ dimensional linear subspace of $\mathbb{R}^n$ when $\mathbf{A}$ is of full rank.

Geometrically, the Grassmann manifold $G_l(\mathbb{R}^n)$ can be conceived as the collection of all the $l$ dimensional subspaces ($l$-planes) of $\mathbb{R}^n$. One can introduce a topology on $G_l(\mathbb{R}^n)$ by defining a metric on it: for arbitrary $\nu,\nu'\in G_l(\mathbb{R}^n)$, the distance between $\nu,\nu'$ can be defined as \cite{mattila}:
\begin{equation}\label{def_dis}
\dist(\nu,\nu'):=\|\mathbf{P}_{\nu}-\mathbf{P}_{\nu'}\|,
\end{equation}
where $\mathbf{P}_{\nu}$ (resp. $\mathbf{P}_{\nu'}$) is the projection matrix onto $\nu$ (resp. $\nu'$), and $\|\cdot\|$ denotes the spectral norm. The Grassmann manifold is then a compact metric space.

We shall next define the coordinates on $G_l(\mathbb{R}^n)$ to introduce its differential manifold structure. Let $F(n,l)$ be the set of all non-degenerate (invertible) $n\times l$ matrices, and let $\sim$ be the following equivalence relation:
If $\mathbf{X},\mathbf{Y}\in F(n,l)$,
then $\mathbf{X}\sim \mathbf{Y}$ if there is an $l\times l$ invertible matrix $\mathbf{V}$ such that $\mathbf{Y = XV}$.
Hence the Grassmann manifold can be defined as a quotient space $G_l(\mathbb{R}^n):=F(n,l)/\sim$, for which we denote by $\pi: F(n,l)\to G_l(\mathbb{R}^n)$ the associated natural projection. For any arbitrary collection of indices $1\le i_1<i_2<\dots<i_l\le n$, let $1\le \bar{i}_1<\bar{i}_2<\dots<\bar{i}_{n-l}\le n$ be the remaining indices.
Given an index set $I=\{i_1,i_2,\dots,i_l\}$, we denote by $\mathbf{X}_I$ the $l\times l$ sub-matrix formed by the rows of $\mathbf{X}$  indexed  by $I$. Define
\begin{align}
V_{I}:=&\{\mathbf{X}\in F(n,l)~|~\det \mathbf{X}_{I}\neq 0\},
\\
U_{I}:=&\pi(U_{I}).\label{e13}
\end{align}
Then $\{U_{I}\}$ constitutes an open covering of $G_l(\mathbb{R}^n)$. For $\mathbf{Y}\in\pi^{-1}(\nu)$,
where $\nu\in U_{I}$, the matrix $\mathbf{X}=\mathbf{Y}\mathbf{Y}_I^{-1}$ is an invariant of $\nu$, meaning that for any other $\tilde{\mathbf{Y}}\in \pi^{-1}(\nu)$, $\tilde{\mathbf{Y}}(\tilde{\mathbf{Y}}_I)^{-1}=\mathbf{X}$. Since $\mathbf{X}_I$ is the $l\times l$ identity matrix, $\mathbf{X}$ is determined by $\mathbf{X}_{I^c}$.
Define $\phi_{I}: U_{I}\to \mathbb{M}(n-l,l), v\mapsto\mathbf{X}_{I^c}$. We call each $(U_{I},\phi_{I})$ a \emph{chart}. Then $\{(U_{I},\phi_{I})~|~1\le i_1<\dots<i_l\le n\}$ forms an \emph{atlas} of $G_l(\mathbb{R}^n)$, meaning that $U_{I}$ covers $G_l(\mathbb{R}^n)$ and any two charts in this collection are \emph{$C^{\infty}$ compatible}.

Concepts such as open sets and interior are well-defined once a topology on $G_{l}(\mathbb{R}^n)$ has been unambiguously chosen. One might notice that there are possibly two topologies defined on $G_{l}(\mathbb{R}^n)$: the metric topology arising from the metric defined in (\ref{def_dis}), and the manifold topology (which is connected to the standard topology on $\mathbb{R}^{ml}$ by all the homeomorphisms $\{\phi_I\}$). Unsurprisingly these two topologies agree, since standard calculations would show that the metric on $U_I$ induced from the Euclidean metric on $\phi_I(U_I)$ is topologically equivalent to the metric defined in (\ref{def_dis}).

Further, since $G_{l}(\mathbb{R}^n)$ is a $C^\infty$ (therefore differentiable) manifold, the concept of measure zero set can be defined as follows:
\begin{definition}\label{def2}\cite[Definition 1.16]{milnor}
A subset $A$ of a differentiable manifold has \emph{measure zero} if $\phi(A\cap U)$ has Lebesgue measure zero for every chart $(U,\phi)$.
\end{definition}
\begin{remark}
The differentiability of the manifold ensures that the definition of the measure zero set is ``consistent'' for the various choices of $\phi$. In particular, to check that a set $A$ is of measure zero, one only needs to pick a collection of homeomorphisms $\{\phi_I\}_{I\in S}$ whose domains cover $U$, and check that $\phi_I(A\cap U_I)$ has measure zero for each $I\in S$.
\end{remark}

The Haar measure, denoted as $\mu$, is the unique probability measure on $G_l(\mathbb{R}^n)$ which is invariant with respect to the orthogonal group's action on $G_l(\mathbb{R}^n)$, that is, the action on the quotient space $G_l(\mathbb{R}^n)$ induced from the action of left multiplication of $n\times n$ orthogonal matrix on $F(n,l)$. The requirement that a set $A$ has zero Haar measure agrees with Definition \ref{def2} \cite{boothby}.
The Haar measure is of practical importance, since it coincides with the probability distribution of the null space of $\mathbf{A}$ when $\mb{A}$ is \emph{rotationally invariant}, that is, the distribution of $\mathbf{A}$ is the same as that of $\mathbf{A}\mb{Q}$ for any orthogonal matrix $\mb{Q}$. In particular, this is true when $\mb{A}$ is a (standard) Gaussian random matrix.

\section{The Relationship between ERC and RRC}\label{secrelation}
\subsection{Equivalence Lost: $\Omega_J^r=\interior(\Omega_J)$}\label{pa}
We have mentioned earlier that NSP is a necessary and sufficient condition for ERC. If $\mathbf{A}\in\mathbb{M}(m,n)$ is in a general position (i.e., the rows of $\mathbf{A}\in\mathbb{M}(m,n)$ are linearly independent), then $\mathbf{A}$ is of full rank, and $\mathcal{N}(\mathbf{A})$ is a $l$-dimensional subspace in $\mathbb{R}^n$ (recall that $l=n-m$). Therefore almost every measurement matrix (except for the set of $\mathbf{A}$'s not in a general position, which is of Lebesgue measure zero) corresponds to an element in $G_{l}(\mathbb{R}^n)$; and this element is sufficient to determine whether NSC, and therefore ERC, is satisfied. By Lemma~\ref{nspcond}, the set of null spaces such that ERC is satisfied is as follows:
\begin{align}\label{exactdef}
\Omega_{J}(n,k,l):=&\{\nu\in G_{l}(\mathbb{R}^n)~|~ J(\mathbf{z}_T)<J(\mathbf{z}_{T^c}),\forall \mathbf{z}\in \nu\setminus \{\mathbf{0}\},T:|T|\le k\}.
\end{align}
For simplicity we shall omit the arguments $n,k,l$ throughout this paper when there is no confusion. If two cost functions induced from the sparseness measures $F,G\in \mathcal{M}$ satisfy the following condition
\begin{equation}\label{e_comp}
\Omega_{J_G}\subseteq\Omega_{J_F},
\end{equation}
then ERC for $G$-minimization implies ERC for $F$-minimization, i.e., $F$ is a better sparseness measure than $G$ in the sense of ERC. In light of this we can describe and compare the performances of different sparseness measures in terms of ERC by a simple set inclusion relation like (\ref{e_comp}).

In Lemma \ref{nspcond}, the necessary and sufficient condition for exact recovery is fully characterized by the structure of the null space. Inspired by this fact we now provide a necessary and sufficient condition for robust recovery:
\begin{theorem}\label{suf2}
Consider the minimization problem in (\ref{minimization2}). Let $\sigma_{\min}$ and $\sigma_{\max}$ be the least and largest singular values of $\mathbf{A}^\top$, respectively. Then RRC holds with constant $C=\frac{2(1+d)}{d\sigma_{\min}}$ if there exists a $d>0$, such that for each $\mathbf{z}\in \mathcal{N}(\mathbf{A})\setminus \{\mathbf{0}\}$, $\mathbf{n}\in\mathbb{R}^n$, $T\subseteq\{1,...,n\}$ satisfying $\|\mathbf{n}\|<d\|\mathbf{z}\|$, and $|T|\le k$, we have the following:
\begin{equation}\label{eq16}
J(\mathbf{z}_T+\mathbf{n}_T)<J(\mathbf{z}_{T^c}+\mathbf{n}_{T^c}).
\end{equation}
Conversely, if RRC holds with $C=\frac{2(1-d)}{d\sigma_{\max}}$ for some $0<d<1$, then $\mathbf{z}\in \mathcal{N}(\mathbf{A})\setminus \{\mathbf{0}\}$, $\mathbf{n}\in\mathbb{R}^n$, $T\subseteq\{1,...,n\}$ satisfying $\|\mathbf{n}\|<d\|\mathbf{z}\|$, and $|T|\le k$, such that (\ref{eq16}) is true.
\end{theorem}
\begin{proof}
See Appendix \ref{p_suf2}.
\end{proof}
\begin{remark}\label{rem1}
As will be clear later in \ref{s3c}, the coefficient
\begin{align}
\bar{d}:=\sup\{d:\textrm{~condition \eqref{eq16} holds}\}
\end{align}
has the geometric interpretation as the distance (measured in the $\sin$ of the principal angle) between $\nu$ and a cone \eqref{cone1} corresponding to $J$; or equivalently, the distance between $\nu$ and the non-ERC set $\Omega_J^c$. Such an interpretation and some asymptotic bounds on $\bar{d}$ will be further explored in Section~\ref{s3c}. Note that $\bar{d}>0$ is equivalent to RRC. As we will show in Counter-example~\ref{ex1} ahead, it is possible for some $(\mb{A},k,J)$ to satisfy ERC, but $\nu$ is ``on the edge'' and $\bar{d}=0$, so that RRC fails.
\end{remark}

An immediate corollary of Theorem~\ref{suf2} is the following, the proof of which is omitted:
\begin{corollary}
Consider the minimization problem in (\ref{minimization2}). The RRC holds \emph{if and only if} there exists a $d>0$, such that for each $\mathbf{z}\in \mathcal{N}(\mathbf{A})\setminus \{\mathbf{0}\}$, $\mathbf{n}\in\mathbb{R}^n$, $T\subseteq\{1,...,n\}$ satisfying $\|\mathbf{n}\|<d\|\mathbf{z}\|$, and $|T|\le k$, we have the following:
\begin{equation}\label{eq166}
J(\mathbf{z}_T+\mathbf{n}_T)<J(\mathbf{z}_{T^c}+\mathbf{n}_{T^c}).
\end{equation}
\end{corollary}

\begin{remark}
RRC easily implies ERC, as can be seen in their definitions (Letting $\mathbf{v}=\mathbf{0}$ in the definition of RRC would result in the definition of the ERC), as well as in Theorem~\ref{suf2} (Letting $\mathbf{n}=\mathbf{0}$).
\end{remark}

From Theorem \ref{suf2} it is clear that the property of robust recovery of a particular matrix is also completely determined by its null space. Moreover, it implies that the subset of $G_{l}(\mathbb{R}^n)$ that guarantees RRC is the following:
\begin{align}\label{robdef}
\Omega^r_J:=&\{\nu\in G_{l}(\mathbb{R}^n)~|~ \exists d>0,\textrm{s.t.}~J(\mathbf{z}_T+\mathbf{n}_T)<J(\mathbf{z}_{T^c}+\mathbf{n}_{T^c}),\forall \mathbf{z}\in \nu\setminus \{\mathbf{0}\}, \mathbf{n}:\|\mathbf{n}\|<d\|\mathbf{z}\|,T:|T|\le k\}.
\end{align}

It is not immediately clear from Lemma~\ref{nspcond} and Theorem~\ref{suf2} the connection between ERC and RRC. However there is a nice relation between these two conditions once taking a perspective from the point set topology:
\begin{theorem}\label{th2}
With the standard topology on $G_{l}(\mathbb{R}^n)$, the following relation holds.
\begin{equation}
\Omega^r_J=\interior(\Omega_J).
\end{equation}
\end{theorem}

\begin{proof}
See Appendix \ref{p_th2}.
\end{proof}
Two questions then arise: are the conditions ERC and RRC equivalent for generic cost functions? If not, how much do they differ from each other?
We shall first address the former question in the remainder of Section~\ref{pa}, while the second question will be discussed in Section~\ref{s3b}. In the special case of $\ell_p$-minimization,
these two conditions are indeed equivalent \cite{lqnsp},
as discussed in the introduction. In view of Theorem \ref{th2},
we can show this result by simply proving that $\Omega_{\ell_p}$ is an open set.

\begin{corollary}\label{th5}
If $0<p\le 1$, then $\Omega_{\ell_p}$ is open, hence $\Omega_{\ell_p}^r=\Omega_{\ell_p}$.
\end{corollary}
\begin{remark}
The statement of $\Omega_{\ell_p}^r=\Omega_{\ell_p}$ itself is essentially ``non-topological'', since it does not involve any topological concepts such as open sets; hence it is interesting that there is a simple topological proof of this result. A comparison of different proof methods can be found in Section~\ref{s5a}.
\end{remark}

Next, we shall show an example in which RRC is strictly stronger than ERC, i.e., $\Omega^{r}_J\varsubsetneqq\Omega_J$.
\begin{counter-example}\label{ex1}
The function
\begin{equation}\label{fex1}
F(t):=t+1-\textrm{e}^{-t}
\end{equation}
defined on $[0,+\infty)$ is a spareness measure. Suppose that $x,y>0,~z=x+y$, $k=1$, and that the null space of the measurement matrix is the following one dimensional linear subspace of $\mathbb{R}^3$
\begin{equation}
\mathcal{N}:=\mathbb{R}(x,y,z)^\top,
\end{equation} Conclusion: in this setting ERC is satisfied, but not RRC.
\end{counter-example}

The cost function in (\ref{fex1}) has two salient properties: strict subadditivity (i.e., $F(x)+F(y)>F(x+y)$ for $x,y>0$) and the existence of a derivative at the origin. In appendix \ref{appex} we shall prove the assertions in the Counter-example using these properties.

\subsection{Equivalence Regained: $\Omega_J\setminus\Omega_J^r$ is zero measure and meagre}\label{s3b}
While strict equivalence of ERC and RRC is lost when passing from $\ell_p$ cost functions to generic sparseness measures, as demonstrated in Counter-example \ref{ex1},
we will show in this subsection that the difference is negligible in some formal senses,
at least for non-decreasing sparseness measures. First we take a closer look at Counter-example \ref{ex1}. Using the subadditivity property and the Taylor expansion of $F$ at the origin,
one can explicitly write out:
\begin{equation}\label{om1}
\Omega_J=\left\{[x_1,x_2,x_3]:2\max_{i=1,2,3}|x_i|\le\sum_{i=1,2,3}|x_i|\right\},
\end{equation}
and
\begin{equation}\label{om2}
\Omega_J^{r}=\left\{[x_1,x_2,x_3]:2\max_{i=1,2,3}|x_i|<\sum_{i=1,2,3}|x_i|\right\}.
\end{equation}
We recall that $\mu$ denotes the Haar measure on $G_{l}(\mathbb{R}^n)$. From (\ref{om1}) and (\ref{om2}) it is intuitively clear in this simple case that $\mu(\Omega_J)=\mu(\Omega_J^{r})$,
i.e. the set of null spaces satisfying ERC and the set of null spaces satisfying RRC differ at most by a set of measure zero.
Recall that the Haar measure agrees with the distribution of the null space of an Gaussian random matrix,
as described in Section~\ref{pc}. This means that if $\mathbf{A}$ is a Gaussian random matrix, then the probability of ERC and RRC are the same,
even though the former is implied by the latter.

The general case tends to be much more complicated. Indeed, there exists an Euclidean set $A$ such that $\mu(\interior(A))<\mu(A)$, so the relation $\Omega^r_J=\interior(\Omega_J)$ alone by no means imply that $\mu(\Omega^r_J)=\mu(\Omega_J)$.
In fact, it is not true in general, as we shall see in Counter-example~\ref{counter2}.

However, the set $\Omega_J\setminus \Omega_J^r$ is still guaranteed to be ``small'' if we assume in addition that $F$ is non-decreasing. The smallness may be described in two distinct senses, namely the measure and Baire category. A measure zero set is of course negligible since its corresponding probability is zero. On the other hand, Baire category has nothing to do with the probability; but it is a purely topological concept, so it's worth pointing out the smallness in this sense given the topic of this paper. A set is said to be of \emph{first category} (or \emph{meagre}) if it is a countable union of nowhere dense sets, which are defined as sets whose closures have empty interiors.
\begin{theorem}\label{probeq}
Suppose $F\in\mathcal{M}$ is a non-decreasing function, then $\overline{\Omega}_J\setminus\interior(\Omega_J)$ is zero measure and of the first category.\footnote{Here the notation ``$\setminus$'' denotes the set minus and $\overline{\Omega}_J$ denotes the closure of $\Omega_J$.}
\end{theorem}

The technical proof of this general result will be given in Appendix~\ref{p_probeq}. In the following we discuss the intuition behind the monotonicity assumption in the theorem through a specific low dimensional example, which does not require any background in measure theory.

We first make some comments on the notations. In the remainder of this subsection we always consider the set $\Omega_J$ associated with a particular $J$, so we shall just write the set as $\Omega$ for brevity. Define the set $\Omega_T$ for each $T:|T|\le k$ as follows:
\begin{align}\label{ot}
\Omega_T:=\{\nu\in G_{l}(\mathbb{R}^n)~|~J(\mathbf{z}_T)<J(\mathbf{z}_{T^c}),\forall \mathbf{z}\in \nu\setminus \{\mathbf{0}\}\},
\end{align}
hence $\Omega=\bigcap_{T:|T|=k}\Omega_T$. Note that this notation is not to be confused with $\Omega_J$ or $\Omega_{\ell_p}$.

Next we shall make some preparatory observations. Notice that
\begin{align}
\overline{\Omega}\setminus\interior(\Omega)
&=\overline{\bigcap_{T:|T|=k}\Omega_T} \setminus \interior(\bigcap_{T:|T|=k}\Omega_T)\nonumber\\
&\subseteq \bigcap_{T:|T|=k}\overline{\Omega}_T\setminus\bigcap_{T:|T|=k}\interior(\Omega_T)\label{explain1}\\
&\subseteq \bigcup_{T:|T|=k}(\overline{\Omega}_T \setminus\interior(\Omega_T)),
\end{align}
where (\ref{explain1}) is because $\overline{\bigcap_{T:|T|=k}\Omega_T}\subseteq \bigcap_{T:|T|=k}\overline{\Omega}_T$ and $\interior(\bigcap_{T:|T|=k}\Omega_T)=\bigcap_{T:|T|=k}\interior(\Omega_T)$. Also, define
\begin{align}\label{defs}
S&=\{\nu\in G_{l}(\mathbb{R}^n)~|~ \textrm{$\forall \mathbf{x}\in \nu$ has at most $l-1$ zero entries}\}\nonumber\\
&=
\bigcap_{
\begin{subarray}{c}
                I\subseteq\{1,\dots,n\}\\
                 |I|=l
               \end{subarray}
}U_I,
\end{align}
where $U_I$ was defined in \eqref{e13}. Clearly $U_I^c$ is of measure zero and of the first category for each $I$, so is $S^c$, a finite union of them. Therefore, we only need to show that $[\overline{\Omega}_T\setminus\interior(\Omega_T)]\cap S$ is of measure zero and of the first category for each $T:|T|\le k$. Since $\phi_{I}$ preserves measure zero sets and the topology, we in turn only need to show that the Euclidean set
\begin{align}\label{27}
\phi_{I}([\overline{\Omega}_T\setminus\interior(\Omega_T)]\cap S)
\end{align}
is of measure zero and of the first category for some fixed $I$ and for every $T:|T|\le k$.

Now we are ready to prove that $\overline{\Omega}_T\setminus\interior(\Omega_T)$ is of measure zero in the special case of $n=3,k=1,m=1,T=\{3\}$, and this proof will demonstrate some basic ideas behind the proof of general case in Appendix~\ref{p_probeq}. It is enough to show that $\phi_I(\overline{\Omega}_T\setminus\interior(\Omega_T)\cap U_I)$ is of measure zero (as a subset of $\mathbb{R}^2$) for $I=\{1,2\}$. For an arbitrary $\nu\in G_2(\mathbb{R}^3)$, define $(a,b):=\phi_I(\nu\in\overline{\Omega}_T\setminus\interior(\Omega_T)\cap U_I)$. Then $\nu$ is the subspace spanned by the columns of the matrix $\left(
                                                 \begin{array}{cc}
                                                   1 & 0 \\
                                                   0 & 1 \\
                                                   a & b \\
                                                 \end{array}
                                               \right)
$, so $\nu\in U_I\setminus\Omega_T$ if and only if
\begin{align}\label{lowdim}
F(x)+F(y)< F(ax+by),~\exists~(x,y)\in \mathbb{R}^2.
\end{align}
If $F$ is a non-decreasing subadditive sparseness measure, from the subadditivity it is easily seen that (\ref{lowdim}) holds if and only if $|a|>1$ or $|b|>1$, and the result easily follows. However we will a prove for the case where $F$ is not necessarily subadditive (while possessing other properties of sparseness measures and being non-increasing), because the idea of this proof will hint on the idea of the general result in theorem \ref{probeq}. By symmetry, we first note that it suffices to consider the region where $a,b\ge0$, in which case $\nu\in U_I\setminus\Omega_T$ if and only if
\begin{align}\label{lowdim1}
F(x)+F(y)< F(ax+by),~\exists~x,y\ge 0.
\end{align}
Let's call the set of all $(a,b)$ satisfying (\ref{lowdim1}) the region A, and its complement in $[0,\infty)^2$ the region B. Then the task is just to show that the boundary between A and B has measure zero (by boundary we mean a point belonging to the closures of both region A and B). This is not always true when $A$ is an arbitrary subset of $[0,\infty)^2$. But since $F$ is non-decreasing, from (\ref{lowdim1}) we deduce the following important property:
\begin{center}
(P) If $(a,b)$ is in region A, then for any $a_+\ge a$ and $b_+\ge b$, the point $(a_+,b_+)$ is also in region~A.
\end{center}
Also, notice that the points $(1,0)$ and $(0,1)$ are on the boundary of $A$ and $B$, and from (P) it is easy to see that the boundary is a subset of $[0,1]^2$. Therefore the boundary of $A$ and $B$ looks like the curve depicted in Figure~\ref{fig1} (but we don't actually need a notion of ``curve'' for this proof.) To measure the area of the boundary, divide $[0,1]^2$ into $m$ rows and $n$ columns uniformly, so that $[0,1]^2$ is covered by small $1/m$ by $1/n$ (closed) rectangles. According to (P), there are only three possibilities concerning the vertices of a rectangle:
\begin{flushleft}
1) Both its upper right and lower left vertices belong to $A$;\\
2) It upper right vertex belongs to $A$ and lower left vertex belongs to $B$;\\
3) Both its upper right and lower left vertices belong to $B$.
\end{flushleft}
Clearly, the union of rectangles of type 1), 2) is a closed set containing $A$, so it also contains $\overline{A}$. Similarly, the union of rectangles of type 2), 3) contains $\overline{B}$. Therefore the boundary set, $\overline{A}\cap\overline{B}$ is contained in the intersection of these two unions, whose measure is total area of type 2) rectangles. However, property (P) implies that the number of type 2) rectangles is at most $m+n$ (one way of seeing this is to note that for any two adjacent columns, the rows for which the rectangles are colored have at most one overlap). Therefore total area of type 2) rectangles is at most $(m+n)/mn$ which converges to zero as $m,n\to\infty$. Thus the measure of $\overline{A}\cap\overline{B}$ must be zero.

Readers familiar with fractal geometry may also realize that (P) implies that $\overline{A}\cap\overline{B}$ is actually a porous set, meaning that there exists $0<\alpha<1$ and $r_0>0$ such that for any $0<r<r_0$ and $(a,b)\in \overline{A}\cap\overline{B}$, there is some $(a',b')\in [0,\infty)^2$ such that the ball centered at $(a',b')$ with radius $\alpha r$ is a subset of the ball centered at $(a,b)$ with radius $r$. In our example we can choose, say, $a'=a+r/100, b'=b+r/100$, and $\alpha=1/200$. A porous Euclidean set is necessarily of measure zero and of the first category. In higher dimensions, the idea of proof is again based on porosity, although the construction is more complicated than in this low dimensional example.
\begin{figure}
  \centering
  \includegraphics[width=3in]{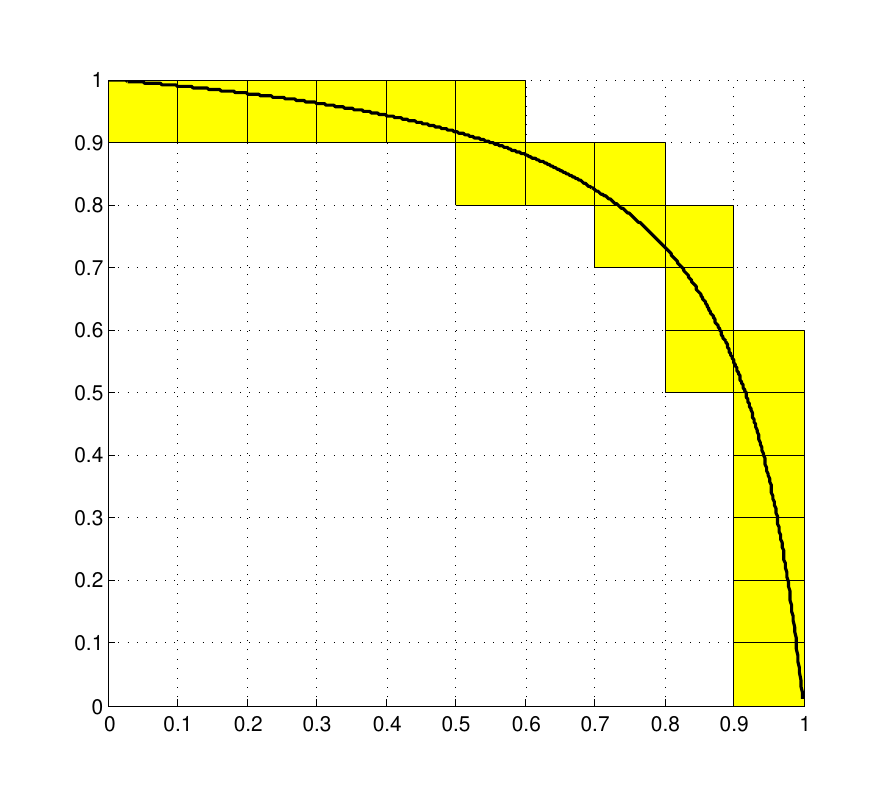}\\
  \caption{The set $[0,1]^2$ is uniformly dived into $10$ columns and $10$ rows. Region A is the region above the curve, and region B the one below it. The type 2) rectangles in the discussion correspond to the colored squares in the figure, the number of which does not exceed $10+10=20$.}\label{fig1}
\end{figure}

Almost all commonly used cost functions that promote sparsity (e.g. the cost functions for $\ell_p$-minimization, ZAP, SCAD, and MCP) satisfy the requirement of $F$ being non-decreasing, so the non-increasing assumption in the theorem is very mild and reasonable. Indeed, it makes sense that a larger nonzero entry should be penalized more in order to promote sparsity.
On the other hand, the non-decreasing requirement is also essential for the validity of Theorem \ref{probeq}. To see this, we construct an example where the ERC set is almost the entire Grassmannian whereas the RRC set is empty.
\begin{counter-example}\label{counter2}
Define
\begin{align}\label{eq_F}
F(x)=\left\{\begin{array}{cc}
      0 & x=0; \\
      0.1 & \textrm{$x>0$ and $x$ is rational};\\
      1 & \textrm{$x>0$ and $x$ is irrational},
    \end{array}
    \right.
\end{align}
and set the dimensions and sparsity to be $m=2,n=3$, and $k=1$. It can be verified that $F$ satisfies the definition of sparseness measure in Definition~\ref{def1}. Then $\mu(\Omega_J)=1$, but $\Omega_J^r=\emptyset$.
\end{counter-example}
\begin{proof}
For arbitrary $x_1,x_2\in \mathbb{R}\setminus \{0\}$, denote by $x_1\simeq x_2$ if the equivalence relation $x_1/x_2\in\mathbb{Q}$ holds\footnote{$\mathbb{Q}$ denotes the set of rational numbers}. Recall the set $S\subseteq G_1(\mathbb{R}^3)$ as defined in (\ref{defs}) is of full measure, and for any $\nu\in S$ and $\mathbf{z}\in \nu\setminus \{\bf 0\}$ we have $z_i\neq 0$, $i=1,2,3$. Then the three coordinates of $\mathbf{z}$ can be partitioned into equivalent classes according to $\simeq$, and the type of partition is independent of the choice of $\mathbf{z}$. Moreover, whether $\nu$ is in $\Omega_J$ or not is completely determined by the type of partition, according to the construction of $F$. For example, we say $\nu$ is of type $(1,1,2)$ if the first two coordinates of $\mathbf{z}$ are from a same equivalence class and the third coordinate is from another equivalence class. From the null space property we can check that the type $(1,2,3)$ is in $\Omega_J$, since for each $\mathbf{z}\in \nu\setminus\{0\}$, there is at least one $i\in T^c$ such that $z_i$ is irrational. However, any type $(1,1,2)$ null space $\nu$ is not in $\Omega_J$, because there exists a $\mathbf{z}\in\nu\setminus\{\bf 0\}$ such that $z_1,z_2$ are rational and $z_3$ is irrational, in which case null space property fails when choosing $T=\{1\}$. Since the null spaces of the type $(1,2,3)$ is of measure $1$, we have that $\mu(\Omega_J)=1$.
On the other hand, since the set of one dimensional subspaces corresponding to the type $(1,1,2)$ is dense in $G_1(\mathbb{R}^3)$ but does not intersect $\Omega_J$, the interior of $\Omega_J$ must be vacuous.
\end{proof}

An immediately corollary of Theorem~\ref{probeq} is that the probability of ERC and RRC are the same for a Gaussian random matrix. More generally, suppose $P$ is the probability measure corresponding to the distribution of the null space of $\mathbf{A}$, and $P$ is absolutely continuous with respect to $\mu$,\footnote{The measure $\mu_1$ is said to be absolutely continuous with respect to the measure $\mu_2$ if $\mu_2(E)=0$ implies $\mu_1(E)=0$, for arbitrary measurable set $E$.} then $P(\Omega_J\setminus\Omega_J^r)=0$.
In particular, it is not counter-intuitive that this should be true if the entries of $\mathbf{A}$ are i.i.d.~generated from a certain continuous distribution (e.g.~sub-Gaussian or subexponential \cite{raskutti2010restricted}), which is a common practice used in generating the observation matrix. Nevertheless, the above speculation requires a formal justification. We formulate this result as a corollary, the proof of which is deferred to Appendix \ref{app_F}.
\begin{corollary}\label{co2}
Suppose $F\in\mathcal{M}$ is a non-decreasing function, and the distribution of the matrix $\mathbf{A}$ is absolutely continuous with respect to the Lebesgue measure on $\mathbb{M}(m,n)$. Then the probability of ERC and RRC are the same. This holds true in particular when $\mathbf{A}$ has i.i.d. entries drawn from a continuous distribution.
\end{corollary}

\begin{remark}
Apart from the one described in Corollary \ref{co2}, another popular method for the generation of $\mathbf{A}$ is by randomly selecting $m$ rows in the $n\times n$ Fourier transform matrix \cite{tao1,tao2}. However in this scheme the probability of ERC and RRC may not agree, since the probability distribution of the null space is discrete rather than continuous on $G_l(\mathbb{R}^n)$. In fact, it is not difficult to construct a random Fourier matrix ensemble and a sparseness measure for which the probability of ERC is strictly larger than that of RRC.
\end{remark}

\subsection{The Quantitative Version}\label{s3c}
Although the characterization in Theorem \ref{th2} is simple and accurate, it fails to convey a quantitative information about robustness: given a subspace in $\Omega_J^r=\interior(\Omega_J)$, we do not know how large the constant $C$ in the definition of RRC is. To obtain such a quantitative information, consider the ``$d$-interior'' of $\Omega_J$, defined as
\begin{align}\label{e32}
d\mbox{-}\interior(\Omega_J):=\{\nu\in G_l(\mathbb{R}^n)|\nu'\in \Omega_J,~\forall~\nu':\dist(\nu,\nu')< d \}.
\end{align}
We remark that by definition, $\interior(\Omega_J)=\bigcup_{d>0} d\mbox{-}\interior(\Omega_J)$. Now a ``quantitative'' version of Theorem \ref{th2} is as follows:
\begin{theorem}\label{th22}
Suppose $\mb{A}$ is of full rank, then $\mathcal{N}(\mathbf{A})\in d\mbox{-}\interior(\Omega_J(n,k,l))$ is equivalent to \eqref{eq16}.
That is, the condition $\mathcal{N}(\mathbf{A})\in d\mbox{-}\interior(\Omega_J(n,k,l))$ implies that RRC is satisfied for $({\bf A},k,J)$ with the robustness constant $C=\frac{2(1+d)}{d\sigma_{\min}}$; and conversely, if $({\bf A},k,J)$ satisfies RRC with $C=\frac{2(1-d)}{d\sigma_{\max}}$ for some $d>0$, then $\mathcal{N}(\mathbf{A})\in d\mbox{-}\interior(\Omega_J(n,k,l))$.
\end{theorem}
\begin{proof}
See Appendix \ref{p_th2}.
\end{proof}
\begin{remark}
A nice feature of the bounds on $C$ in Theorem~\ref{th22} is that $d$ only depends on $\mathcal{N}(\mb{A})$ whereas the information of $\mb{A}$ only enters the bounds through the extremal values $\sigma_{\max}$ and $\sigma_{\min}$. If the measurement matrix can be designed, then the designer can always perform a QR transform on $\mb{A}^{\top}$ to obtain a new measurement matrix with $\sigma_{\max}=\sigma_{\min}$ while the null space is unchanged. Clearly, in such a case the upper and lower bounds on $C$ in Theorem~\ref{th22} are very tight for small $d$.
\end{remark}

In principle, the supremum of $d$ such that $\mathcal{N}({\bf A})\in d\mbox{-}\interior(\Omega_J)$ is completely determined by $({\bf A},k,J)$. However, exactly computing $d$ for a given $\bf A$ seems to be out of reach since $T$ may take $\left(
                                   \begin{array}{c}
                                     n \\
                                     k \\
                                   \end{array}
                                 \right)
$ number of values. Two practical approaches of estimating $d$ will be discussed in Section~\ref{sec5}.

\section{Comparison of Different Sparseness Measures}\label{secrules}
In this section we provide some methods to compare the performance between two sparseness measures in terms of ERC or RRC. Since $\Omega_J$, $\interior(\Omega_J)$, and $d\mbox{-}\interior(\Omega_J)$ are shown to correspond to the measurement matrices satisfying ERC, RRC, or with a particular robustness constant, it's easy to compare the performances of two sparseness measures if an inclusion relation such as $\Omega_{J_1}\subseteq\Omega_{J_2}$ is available. However, sometimes its not true that $\Omega_{J_1}\subseteq\Omega_{J_2}$, but it may still be possible to show $\Omega_{J_1}\subseteq\overline{\Omega}_{J_2}$. The second relation is not terribly different than the first one, since we have shown that $\overline{\Omega}_{J_2}\setminus\Omega_{J_2}$ is negligible when $F_2$ is non-decreasing. Therefore, both the topological characterization of RRC and the probabilistic (measure-theoretic) viewpoint become particularly useful when passing from the $\ell_p$ cost functions to general sparseness measures.

The following lemma comes from the corresponding result for ERC in \cite{Gribonval} and our interior point characterization of RRC:
\begin{lemma}\label{le2}
Suppose $F,G\in\mathcal{M}$. If $F,G$ are non-decreasing and $F/G$ is non-increasing on $\mathbb{R}^+$, then we have $\Omega_{J_G}\subseteq\Omega_{J_F}$ and $\Omega^{r}_{J_G}\subseteq\Omega^{r}_{J_F}$.
\end{lemma}
\begin{proof}
The fact that $\Omega_{J_G}\subseteq\Omega_{J_F}$ comes from \cite[Lemma 7]{Gribonval}. It then follows that $\Omega^{r}_{J_G}\subseteq\Omega^{r}_{J_F}$ from Theorem \ref{th2}.
\end{proof}

The set inclusions formulas in Lemma~\ref{le2} means that the sparseness measure $F$ is better than $G$, in the sense that whenever the cost function $J_G$ guarantees ERC/RRC, so does the $J_F$. By letting $G(x):=x^q$ in this lemma we can obtain the following result:
\begin{corollary}\label{co1}
Suppose $F\in\mathcal{M},~p\in(0,1]$. If $F$ is non-decreasing and $F(x)/x^p$ is non-increasing on $\mathbb{R}^+$, then we have $\Omega_{\ell_p}\subseteq\Omega_{J_F}$ and $\Omega^{r}_{\ell_p}\subseteq\Omega^{r}_{J_F}$.
\end{corollary}

Corollary \ref{co1} gives a condition such that $J_F$ is better than $\ell_p$ in the sense of ERC and RRC. Conversely, we shall show that the asymptotic of $F$ around $0^+$ and $+\infty$ gives a sufficient condition that $\ell_p$ is better than $J_F$ in terms of probability.

The following result implies that, in some sense, it's not good to design an $F$ which is differentiable (or Holder continuous) at zero or infinity, as far as the worst case performance is concerned:
\begin{theorem}\label{th1}
Suppose $F\in\mathcal{M},~p\in(0,1]$. If $\lim_{x\to 0^+}F(x)/x^p$ or $\lim_{x\to\infty}F(x)/x^p$ exist and is positive, then $\Omega_{J_F}\subseteq\overline{\Omega}_{\ell_p}$, and $\mu(\Omega_{J_F})\le \mu(\Omega_{\ell_p})$.
\end{theorem}
\begin{proof}
See Appendix \ref{ap2}.
\end{proof}

We Remark that $\mu(\Omega_{J_F})\le \mu(\Omega_{\ell_p})$ in Theorem \ref{th1} cannot be replaced by the stronger set inclusion relation $\Omega_{J_F}\subseteq\Omega_{\ell_p}$, which holds for $\ell_p$ cost functions but fails for general sparseness measures. Thus the measure-theoretic viewpoint allows us to restore a comparison criteria when extending $\ell_p$-minimization to the $F$-minimization.

As an illustration, we demonstrate how to derive the relation between ZAP \cite{zap} and $\ell_1$-minimization from above results. Consider the typical form of sparseness measure used in the ZAP algorithm (which is essentially the same as the minimax concave penalty (MCP) \cite{MCP} familiar to statisticians):
\begin{align}\label{MCPP}
F(x)=
\left\{
\begin{array}{cc}
   \alpha x-\alpha^2x^2 & x<1/\alpha;\\
  1 & \textrm{otherwise},\end{array}
\right.\end{align}
where the tuning parameter $\alpha$ is usually chosen as the inverse of the standard deviation of the non-zero entries in $\mathbf{\bar{x}}$. Our following result says that, while ZAP performs far better than $\ell_1$-minimization in the average case, as shown in the numerical experiments \cite{zap}, the worst case performance (requiring all sparse vectors can be constructed) of the two cost functions are the same:
\begin{corollary}\label{co_zap}
If $F$ is non-decreasing and $F(x)/x$ is non-increasing on $\mathbb{R}^+$, then
\begin{equation}
\mu(\Omega_J)\ge\mu(\Omega_{\ell_1}).
\end{equation}
Moreover, the equally is achieved if in addition $\lim_{x\to 0^+}F(x)/x<\infty$, which is true for ZAP or MCP.
\end{corollary}
\begin{proof}
Using Corollary \ref{co1} and Theorem \ref{th1} with $p=1$ one can obtain both the lower and upper bound on $\mu(\Omega_J)$ respectively.
\end{proof}

The result of Corollary \ref{co_zap} is not in contradiction with the proverbial fact that concave penalties induce smaller risk, since we are different benchmarks of performance. When the parameter $\alpha$ can be tuned according to the statistics of the variables, concave penalties usually have better average performance (risk); but this is irrelevant to our worst case analysis.

An immediate consequence of Corollary~\ref{co1} is that
\begin{align}
d\mbox{-}\interior(\Omega_{\ell_1})\subseteq d\mbox{-}\interior(\Omega_J)
\end{align}
for any $d>0$ if $F$ is non-decreasing and $F(x)/x$ is non-increasing on $\mathbb{R}^+$. Since the proof of Theorem~\ref{thm7} is based on lower bounding $\mu(d\mbox{-}\Omega_{\ell_1})$, we immediately obtain:
\begin{corollary}\label{cor8}
The tradeoffs between the sampling rate and robustness described in Theorem~\ref{thm7}, Corollary~\ref{as} and Corollary~\ref{as1} are also achievable for $F$-minimization if $F$ is non-decreasing and $F(x)/x$ is non-increasing on $\mathbb{R}^+$.
\end{corollary}

We end this section by summarizing the relationship between the various requirements on $F$ appeared in this section:
\begin{proposition}\label{prop1}
Assuming that $0\le p\le1$, $F:[0,+\infty)\to[0,+\infty)$, and $F(0)=0$, we have\\
(1) $F$ is concave $\Longrightarrow$ $F(t)/t$ is non-increasing;\\
(2) $F(t)/t^p$ is non-increasing $\Longrightarrow$ $F(t)/t$ is non-increasing;\\
(3) $F(t)/t$ is non-increasing $\Longrightarrow$ $F$ is subadditive.
\end{proposition}
\begin{proof}
See Appendix \ref{prop1_proof}.
\end{proof}

\section{Applications: Tradeoffs between Robustness and Measurement Rate}\label{sec5}
We now apply the results in previous sections to prove achievable tradeoffs between robustness and measurement rate in various settings.
\subsection{Escape through the Mesh}\label{sgordon}
By Theorem~\ref{th22}, the robustness constant $C$ can be characterized by the largest $d$ such that $\mb{A}\in d\mbox{-}\interior(\Omega_J)$, therefore we are interested in the probability $\mathbb{P}[\mathcal{N}(\mb{A})\in d\mbox{-}\interior(\Omega_J)]$ associated with a random matrix. This can be well estimated in the case of a rotationally invariant matrix, for which $\mathcal{N}(\mb{A})$ is uniformly distributed on the Grassmannian. The key idea of our analysis is Gordon's \emph{escape through the mesh} theorem \cite{gordon}, which was employed in the study of exact reconstruction of sparse signals via $\ell_1$-minimization by Rudelson and Vershynin \cite{rudelson}. With some additional observations, we can use this approach to bound the robustness constant $C$. Define the sets of vectors
\begin{align}
\mathcal{D}_J(n,k):=\{\mathbf{z}\in \mathbb{R}^n\setminus \{\mathbf{0}\}|
J(\mathbf{z}_T)\ge J(\mathbf{z}_{T^c}),~\exists
T\subseteq\{1,...,n\}:|T|\le k\},
\end{align}
\begin{align}
\mathcal{D}_{J,d}(n,k):&=\{\mathbf{z}\in \mathbb{R}^n\setminus \{\mathbf{0}\}|
J(\mathbf{z}_T+\mathbf{n}_T)\ge J(\mathbf{z}_{T^c}+\mathbf{n}_{T^c}),
~\exists
\mathbf{n}\in\mathbb{R}^n:\|\mathbf{n}\|<d\|\mathbf{z}\|,
T\subseteq\{1,...,n\}:|T|\le k\}\\
&=\{\mathbf{z}\in \mathbb{R}^n\setminus \{\mathbf{0}\}|
\mathbf{z}+\mathbf{n}\in\mathcal{D}_d(n,k),
~\exists
\mathbf{n}\in\mathbb{R}^n:\|\mathbf{n}\|<d\|\mathbf{z}\|,
T\subseteq\{1,...,n\}:|T|\le k\}.\label{33}
\end{align}
Thus $\mathcal{D}_{J,d}(n,k)$ can be seen as a ``robust'' or ``extended'' version of $\mathcal{D}_J(n,k)$. Again, we shall omit the subscript $J$ when there is no confusion from the context. Define the cones
\begin{align}\label{cone1}
\mathcal{C}(n,k)&:=\{\mathbf{x}\in\mathbb{R}^n~|~\exists t\in \mathbb{R}~{\rm s.t.}~t\mathbf{x}\in \mathcal{D}(n,k)\}
\end{align}
\begin{align}
\mathcal{C}_d(n,k):=\{\mathbf{x}\in\mathbb{R}^n~|~\exists t\in \mathbb{R}~{\rm s.t.}~t\mathbf{x}\in \mathcal{D}_d(n,k)\};\label{def4}
\end{align}
 Also define the following subsets of unit sphere in $\mathbb{R}^n$:
\begin{equation}\label{e43}
\mathcal{K}(n,k):=\mathcal{C}(n,k)\cap S^{n-1},
\end{equation}
\begin{equation}
\mathcal{K}_d(n,k):=\mathcal{C}_d(n,k)\cap S^{n-1}.
\end{equation}
Then by definitions it is easy to see that
\begin{align}
\Omega_J&=\{\nu\in G_l(\mathbb{R}^n)|\mathcal{K}(n,k)\cap\nu=\emptyset\},\label{3940}
\end{align}
and Lemma~\ref{lmth2} in Appendix~\ref{p_th2} implies that
\begin{align}
d\mbox{-}\interior(\Omega_J)&=\{\nu\in G_l(\mathbb{R}^n)|\mathcal{K}_d(n,k)\cap\nu=\emptyset\}\label{e41}.
\end{align}
For any vector $\mathbf{g}\in \mathbb{R}^n$ and $\epsilon>0$, there exists $\mathbf{x}'\in \mathcal{K}_d(n,k)$ so that
\begin{align}\label{ineq1}
\sup_{\mathbf{x} \in \mathcal{K}_d(n,k)}\mathbf{g}^\top\mathbf{x}
\le \mathbf{g}^\top\mathbf{x}'+\epsilon
\end{align}
By (\ref{def4}) and (\ref{33}), there exists $t\neq 0$ and $\mathbf{n}':\|\mathbf{n}'\|<d\|t\mathbf{x}'\|$ such that $t\mathbf{x}'+\mathbf{n}'\in \mathcal{D}(n,k)$. Therefore $\mathbf{x}'+\mathbf{n}\in \mathcal{C}(n,k)$ where $\mathbf{n}:=t^{-1}\mathbf{n}'$. Let $\mathbf{y}$ be the projection of $\mathbf{x}'$ onto the one dimensional subspace spanned by $\mathbf{x}'+\mathbf{n}$. Then $\mathbf{y}\in \mathcal{C}(n,k)$ because $\mathcal{C}(n,k)$ is a cone. Also $\|\mathbf{y}-\mathbf{x}'\|\le \|{\bf (x'+n)-x'}\|\le d\|{\bf x'}\|=d$, $\|\mathbf{y}\|\le 1$ by properties of the projection. Thus
\begin{align}
\mathbf{g}^\top\mathbf{x}'&=\mathbf{g}^\top(\mathbf{x}'-\mathbf{y})+\mathbf{g}^\top\mathbf{y}\\
&\le \|\mathbf{g}\|\|\mathbf{x}'-\mathbf{y}\|+|\mathbf{g}^\top\mathbf{y}/\|\mathbf{y}\||\cdot\|\mathbf{y}\|\\
&\le d\|\mathbf{g}\|+|\mathbf{g}^\top\mathbf{y}/\|\mathbf{y}\||\\
&\le d\|\mathbf{g}\|+\sup_{\mathbf{x} \in \mathcal{K}(n,k)}\mathbf{g}^\top\mathbf{x}\label{ineq2}
\end{align}
The last inequality used the fact that $\pm\mathbf{y}/\|\mathbf{y}\|\in\mathcal{C}(n,k)$, since $\mathcal{C}(n,k)$ is centrally symmetric. Now (\ref{ineq1}), (\ref{ineq2}) and the arbitrariness of $\epsilon$ give
\begin{equation}\label{ineq3}
\sup_{\mathbf{x} \in \mathcal{K}_d(n,k)}\mathbf{g}^\top\mathbf{x}
\le d\|\mathbf{g}\|+\sup_{\mathbf{x} \in \mathcal{K}(n,k)}\mathbf{g}^\top\mathbf{x}
\end{equation}
This result will be useful soon in connecting the two sets $\mathcal{K}_d(n,k)$ and $\mathcal{K}(n,k)$.
\begin{definition}
The \emph{Gaussian width} of a subset of $\mathcal{K}\subseteq S^{n-1}$ is defined as
\begin{equation}
w(\mathcal{K})=\mathbb{E}\sup_{x \in \mathcal{K}}\mathbf{g}^\top\mathbf{x}
\end{equation}
where $\mathbf{g}\sim\mathcal{N}(\mb{0},\mb{I}_n)$.
\end{definition}

From (\ref{ineq3}), the Gaussian width of the extended set $\mathcal{K}_d(n,k)$ is upper bounded by
\begin{align}\label{48}
w(\mathcal{K}_d(n,k))
\le& w(\mathcal{K}(n,k))+d\mathbb{E}\|\mathbf{n}\|
\\
\le& w(\mathcal{K}(n,k))+d\sqrt{\mathbb{E}\|\mathbf{n}\|^2}
\\
\le& w(\mathcal{K}(n,k))+d\sqrt{n}.
\end{align}
A small Gaussian width implies that a random linear subspace of $\mathbb{R}^n$ is not likely to intersect with it:
\begin{theorem}[Escape Through the Mesh (Gordon) \cite{gordon}]\label{thm5}
Let $\mathcal{K}$ be a subset of the unit Euclidean sphere $S^{n-1}$ in $R^n$. Let $\nu$ be a random $(n-m)$-dimensional subspace of $R^n$, distributed uniformly in the Grassmannian with respect to the Haar measure. Assume that
\begin{equation}\label{wcond}
w(\mathcal{K})<\sqrt{m}.
\end{equation}
Then $\nu\cap \mathcal{K}=\emptyset$ with probability at least
\begin{equation}
1-2.5\exp\left(-\frac{(m/\sqrt{m+1}-w(\mathcal{K}))^2}{18}\right).
\end{equation}
\end{theorem}
\begin{remark}
As noted in \cite{rudelson}, the original coefficient $3.5$ in \cite{gordon} can be replaced with $2.5$ shown above. Nevertheless, its exact value does not matter for our purposes.
\end{remark}

From (\ref{3940}), \eqref{e41}, (\ref{48}) and Theorem~\ref{thm5}, one immediately obtains the following estimate of the probability of $\Omega_J$ and $d\mbox{-}\interior(\Omega_J)$ in the case of Gaussian measurement matrix:
\begin{theorem}\label{k}
If $w(\mathcal{K}_J(n,k))<\sqrt{m}-d\sqrt{n}$, then
\begin{equation}
\mu(d\mbox{-}\interior(\Omega_J))\ge 1-2.5\exp\left(-\frac{(m/\sqrt{m+1}-w(\mathcal{K}_J(n,k))-d\sqrt{n})^2}{18}\right).
\end{equation}
\end{theorem}
Note that in Theorem \ref{k}, the results rely on the Gaussian width $\mathcal{K}_J(n,k)$, which is essentially determined by $J$. In the remainder of this subsection we shall analyse the case where $J$ is the $\ell_1$ norm, which then applies to all $F$ satisfying the assumptions in Corollary~\ref{co1} with $p=1$. As remarked earlier, the asymptotic analysis of $\mathcal{K}_{\ell_1}(n,k)$ was carried out by \cite{rudelson} in the study of exact recovery property. Lemma 4.4 and 4.5 in \cite{rudelson} combined yield the following upper bound on the Gaussian width of $\mathcal{K}_{\ell_1}(n,k)$:
\begin{equation}
w(\mathcal{K}_{\ell_1}(n,k))\le 2\sqrt{k(3+2\log(n/k))}\cdot\zeta(n,k)
\end{equation}
where
\begin{equation}
\zeta(n,k)=\exp\left(\frac{\log(1+2\log(en/k))}{4\log(en/k)}+\frac{1}{24k^2\log(en/k)}\right).
\end{equation}
we shall consider the asymptotic case where $k,n,m$ scales linearly, i.e.,
\begin{align}\label{grow1}
n=\lfloor\beta k\rfloor
\end{align}
and
\begin{align}\label{grow2}
m=\lceil\gamma k\rceil
\end{align}
for some constants $\beta>\gamma\ge1$.
Then $w(\mathcal{K}_{\ell_1}(n,k))$ satisfies the condition $w(\mathcal{K}_{\ell_1})<\sqrt{m}-d\sqrt{n}$ in Theorem \ref{k} for large $k,n,m$ if the scaling parameters $\beta,\gamma$ and the number $d$ satisfy
\begin{equation}
2\sqrt{(3+2\log(\beta))}\cdot\exp\left(\frac{\log(1+2\log(e\beta))}{4\log(e\beta)}\right)<\sqrt{\gamma}-d\sqrt{\beta}.
\end{equation}
Define
\begin{equation}\label{delta}
\delta(\beta,\gamma)=\frac{1}{\sqrt{\beta}}\left(\sqrt{\gamma}-
2\sqrt{(3+2\log(\beta))}\cdot\exp\left(\frac{\log(1+2\log(e\beta))}{4\log(e\beta)}\right)\right).
\end{equation}
Notice that $\delta(\beta,\gamma)>0$ when $\gamma>4(3+2\log(\beta))\cdot\exp\left(\frac{\log(1+2\log(e\beta))}{2\log(e\beta)}\right)$. If this is the case, the escape through the mesh theorem implies that $\mu(d\mbox{-}\interior(\Omega_J))$ tends to one as $k\to \infty$, if $d<\delta(\beta,\gamma)$. Therefore we have
\begin{theorem}\label{thm7}
Suppose $n=\lfloor\beta k\rfloor$, $m=\lceil\gamma k\rceil$ for some constants $\beta>\gamma\ge1$, and $\mb{A}$ is rotationally invariant. Assume that $F$ satisfy the condition in Corollary~\ref{co1} with $p=1$. If $\delta(\beta,\gamma)$ defined in (\ref{delta}) is positive and $d<\delta(\beta,\gamma)$, then for $k$ large enough,
the $F$-minimization satisfies RRC with the robustness constant $C=\frac{2(1+d)}{d\sigma_{\min}}$
with probability exceeding
\begin{align}
1-2.5\exp\left(-\frac{n(\delta(\beta,\gamma)-d)^2}{18}\right).
\end{align}
\end{theorem}

The bound on $C$ in Theorem~\ref{thm7} is random since $\sigma_{\min}$ depends on the random matrix $\bf A$. We can particularize Theorem~\ref{thm7} to some rotationally invariant matrix ensembles to obtain convergence results.
Random matrix theory (see for example \cite{DC} and the references therein) reveals that if the entries of $\mathbf{A}$ are i.i.d.~Gaussian with zero mean and variance $1/n$, then $\sigma_{\min}(\mathbf{A}^\top)$ converges to $1-\sqrt{\gamma/\beta}$ almost surely as $k\to\infty$. Thus by Theorem \ref{th22}, we have:
\begin{corollary}[standard Gaussian ensemble]\label{as}
Suppose $n=\lfloor\beta k\rfloor$, $m=\lceil\gamma k\rceil$ for some constants $\beta>\gamma\ge1$, and the entries of the measurement matrix are i.i.d. Gaussian with zero mean and variance $1/n$. Assume that $F$ satisfy the condition in Corollary~\ref{co1} with $p=1$. If $\delta(\beta,\gamma)$ defined in (\ref{delta}) is positive, then with probability converging to one as $k\to\infty$, the $F$-minimization satisfies RRC with the robustness constant $C=\frac{2(1+\delta(\beta,\gamma))}{\delta(\beta,\gamma)(1-\sqrt{\gamma/\beta})}$, where $\delta(\beta,\gamma)$ is defined in (\ref{delta}).
\end{corollary}

The proof of Corollary \ref{as} follows directly from the preceding discussion. Notice that $\beta$ characterizes the sparsity, which is determined by the nature of the signal; and $\frac{\gamma}{\beta}$ is the measurement rate, which is may be controlled by the designer. If we view $\beta$ as a fixed parameter, then Corollary \ref{as} can be interpreted as the tradeoff between measurement rate $\frac{\gamma}{\beta}$ and robustness $C$. Moreover there is phase transition point for $\gamma$ above which $\ell_1$-minimization becomes robust, and hence also guarantees exact recovery.


Another rotationally invariant measurement matrix can be obtained by selecting $m$ rows from a matrix in the circular real ensemble (CRE($n$)) \cite{dyson1962threefold}, which has the uniform distribution on the set of orthogonal matrices of dimension $n$. An advantage of such a measurement matrix construction over the Gaussian ensemble is that $\sigma_{\max}=\sigma_{\min}=1$ which does not introduce extra slackness into the upper and lower bounds on $C$ in Theorem~\ref{th22}.

\begin{corollary}[circular real ensemble]\label{as1}
Suppose $n=\lfloor\beta k\rfloor$, $m=\lceil\gamma k\rceil$ for some constants $\beta>\gamma\ge1$, and $\mb{A}$ is composed of the first $m$ rows of a matrix in CRE($n$). Assume that $F$ satisfy the condition in Corollary~\ref{co1} with $p=1$. If $\delta(\beta,\gamma)$ defined in (\ref{delta}) is positive, then with probability converging to one as $k\to\infty$, the $F$-minimization satisfies RRC with the robustness constant $C=\frac{2(1+\delta(\beta,\gamma))}{\delta(\beta,\gamma)}$, where $\delta(\beta,\gamma)$ is defined in (\ref{delta}).
\end{corollary}

\begin{remark}\label{rem8}
The derivations in this subsection relies on the Gaussian width associated with the $\ell_1$ penalty function. It is possible to extend the approach to other cost functions as long as an estimate of the associated Gaussian width is available.
For example, in \cite[Section 3]{stojnic}, an upper bound on $\mathcal{K}_{\ell_p}$ was derived, although numerical optimizations needs to be solved in order to compute that bound. A recent work \cite{amelunxen2014living} also evaluated a related quantity called statistical dimension for various convex cones.
\end{remark}

\subsection{Beyond Rotationally Matrix Ensembles}\label{sbeyond}
The approach in \ref{sgordon} based on Gordon's escape through the mesh theorem relies on the uniformity of the distribution of $\mathcal{N}(\mb{A})$. Without uniformity, it may still be possible to upper bound $C$ using our quantitative characterization of RRC (Theorem~\ref{th22}) and a connection between NSP and a \emph{restricted eigenvalue condition}. In this subsection we illustrate this alternative approach in the case of sub-Gaussian random matrices, while in principle this method is applicable to any random matrix satisfying the \emph{restricted isometry property}.

We begin by defining a version of \emph{restricted eigenvalue condition}, which is a generalization of the definitions in \cite{bickel2009simultaneous}\cite{raskutti2010restricted}.
\begin{definition}
We say ${\sf RE}_J(k,c)$ is satisfied if
\begin{align}
\|{\bf Ax}\|\ge c \textrm{ for all }{\bf x}\in\mathcal{K}(n,k)
\end{align}
where $\mathcal{K}(n,k)$ is as defined in \eqref{e43}.
\end{definition}
\begin{lemma}\label{lem5}
If ${\sf RE}_J(k,c)$ is satisfied and $c-d\sigma_{\max}\ge0$, then $\mathcal{N}({\bf A})\in d\mbox{-}\interior(\Omega_J)$.
\end{lemma}
\begin{proof}
For any ${\bf x}\in \mathcal{K}(n,k)$ and ${\bf v}$ satisfying $\|{\bf v}\|<d$, we have
\begin{align}
\|{\bf A(x+v)}\|&\ge\|\bf Ax\|-\|Av\|
\\
&>c-d\sigma_{\max}
\\
&\ge0.
\end{align}
Hence ${\bf x+v}\notin \mathcal{N}(\bf A)$. This shows that $\sin(\angle({\bf x}_1,{\bf x}_2))\ge d$ for any ${\bf x}_1\in\mathcal{D}_J(n,k)$ and ${\bf x}_2\in \mathcal{N}(\bf A)$, and hence $\mathcal{N}({\bf A})\in d\mbox{-}\interior(\Omega_J)$ by the first part of Theorem~\ref{th22}.
\end{proof}

\begin{definition}[Restricted isometry propery] Given an $m\times n$ matrix $\bf A$, define
\begin{align}
\phi_{\min}(k):=\min_{\mb{x}\in\mathbb{R}^n\setminus\{\mb{0}\},|\supp(\mb{x})|\le k}\frac{\|{\bf Ax}\|^2}{\|{\bf x}\|^2}
\end{align}
and
\begin{align}
\phi_{\max}(k):=\max_{\mb{x}\in\mathbb{R}^n\setminus\{\mb{0}\},|\supp(\mb{x})|\le k}\frac{\|{\bf Ax}\|^2}{\|{\bf x}\|^2}.
\end{align}
If
\begin{align}
\phi_{\min}(k)\le 1 \le \phi_{\max}(k)
\end{align}
then $\bf A$ is said to satisfy ${\sf RIP}(k,\delta)$ for all $\delta\le\max\{1-\phi_{\min}(k),\phi_{\max}(k)-1\}$.
\end{definition}

When $J$ is the $\ell_p$ norm, the restricted eigenvalue condition is related to RIP via the following, which can be seen as the $\ell_p$ version of [Lemma~3]\cite{bickel2009simultaneous}:
\begin{lemma}\label{lem6}
If
\begin{align}
c:=\frac{\sqrt{\phi_{\min}(b+k)}-\sqrt{\phi_{\max}(b)}\left(\frac{k}{b}\right)^{\frac{1}{p}-\frac{1}{2}}}
{1+\left(\frac{k}{b}\right)^{\frac{1}{p}-\frac{1}{2}}}>0
\end{align}
for some $k,b$ such that $1\le k\le \frac{n}{2}$, $b>k$, $k+b\le n$ and $0<p\le1$, then
${\sf RE}_{\ell_p}(k,c)$ is satisfied.
\end{lemma}
\begin{proof}
See Appendix~\ref{lem6proof}.
\end{proof}

Sub-Gaussian random variables are commonly defined as follows. For equivalent definitions, see \cite[Lemma~5.5]{vershynin2010introduction}.
\begin{definition}\label{def8}
A random variable $X$ is said to be \emph{sub-Gaussian with variance proxy $\sigma^2$} if for all $t\in\mathbb{R}$,
\begin{align}
\mathbb{E}\exp(tX)\le\exp(\sigma^2t^2).
\end{align}
\end{definition}

Note that Definition~\ref{def8} implies that $\mathbb{E}X=\left.\frac{{\rm d}}{{\rm d}t}\mathbb{E}\exp(tX)\right|_{t=0}=0$.

\begin{definition}[Isotropic random vectors]\label{def9}
A random vector $\mb{X}\in\mathbb{R}^n$ is called \emph{isotropic} if $\mathbb{E}\mb{X}\mb{X}^{\top}=t\mb{I}$ for some $t>0$.
\end{definition}

Note that Definition~\ref{def9} only imposes a second moment condition on $\mb{X}$ and does not imply that $\frac{\mb{X}}{\|{\mb{X}}\|}$ is uniformly distributed on the unit sphere. According to Definition~\ref{def9}, a random matrix with i.i.d.~entries has independent isotropic rows. Thus results on RIP for isotropic sub-Gaussian matrices apply. The following result is a specialization of \cite[Theorem~5.65]{vershynin2010introduction} to i.i.d.~sub-Gaussian matrices:
\begin{lemma}[RIP for sub-Gaussian matrix]\cite{vershynin2010introduction}\label{lem7}
Let $\bf B$ be a random matrix with i.i.d.~entries having a fixed sub-Gaussian distribution, and $\mb{A}=\frac{1}{\sqrt{m}}\mb{B}$. Then for every $1\le k\le n$, $0<\delta<1$ and
\begin{align}
m\ge D\delta^{-2}k\log\frac{en}{k},
\end{align}
the matrix $\bf A$ satisfies ${\sf RIP}(k,\delta)$ with probability at least $1-2\exp(-E\delta^2m)$, where the constants $D,E>0$ depend only on the variance proxy of the sub-Gaussian distribution.
\end{lemma}
\begin{remark}
The original definition of sub-Gaussian random variable in \cite{vershynin2010introduction} is equivalent to our definition up to a universal constant, see \cite[Lemma~5.5]{vershynin2010introduction}. If necessary, an upper bound on the constant $D$ can be obtained by tracking down the proof of \cite[Theorem~5.65]{vershynin2010introduction} and the arguments therein.
\end{remark}

\begin{theorem}
Suppose $F\in\mathcal{M}$, $F$ is non-decreasing and $F(x)/x^p$ is non-increasing on $\mathbb{R}^+$ for some $p\in(0,1]$. The random matrix $\bf B$ has i.i.d.~sub-Gaussian entries with unit variance and $\mb{A}:=\frac{1}{\sqrt{m}}\mb{B}$. Assume the linear growth regime \eqref{grow1} and \eqref{grow2} with $\beta>2$. If
\begin{align}
c:=\max_{\frac{1}{\beta-1}\le \rho\le 1}
\frac{\sqrt{1-\sqrt{\frac{D(1+\rho)}{\gamma\rho}\log\frac{e\beta\rho}{1+\rho}}}
-\sqrt{1+\sqrt{\frac{D\log(e\rho\beta)}{\gamma\rho}}}\rho^{\frac{1}{p}-\frac{1}{2}}}
{1+\rho^{\frac{1}{p}-\frac{1}{2}}}>0
\end{align}
where $D$ is the constant in Lemma~\ref{lem7} depending only on the variance proxy of the sub-Gaussian distribution,
then with probability converging to one as $k\to\infty$, $J$-minimization satisfies RRC with the robustness constant
\begin{align}
C=\frac{2(1+c+\sqrt{\beta/\gamma})}{c(\sqrt{\beta/\gamma}-1)}.
\end{align}
\end{theorem}
\begin{proof}
The claim follows directly from Lemma~\ref{lem5}, Lemma~\ref{lem6}, Lemma~\ref{lem7}, Corollary~\ref{co1}, Theorem~\ref{th22} and the fact that
\begin{align}
\sigma_{\max}\to \sqrt{\frac{\beta}{\gamma}}+1,
\\
\sigma_{\min}\to \sqrt{\frac{\beta}{\gamma}}-1,
\end{align}
as $k\to\infty$ almost surely due to Bai-Yin’s law \cite{bai1993limit}.
\end{proof}

Clearly, the line of arguments above can be applied to other matrix ensembles by replacing Lemma~\ref{lem7} with the RIP result for the corresponding ensemble; for example, see \cite{adamczak2011restricted} for the sub-exponential ensemble and \cite{rauhut2010compressive} for the Fourier ensemble.

\section{Comparison with Other Works}\label{comp}
\subsection{The ERC/RRC Equivalence for $\ell_p$-minimization}\label{s5a}
To the best of our knowledge, the \emph{exact} characterization of robustness of $\ell_p$-minimization first appeared in \cite{Foucart}, where the definition of robustness is the same as in our paper. In \cite{Foucart} a variant of the null space property, called NSP',
was proposed as a sufficient condition for the robustness of $\ell_p$ minimization. The NSP' is obviously stronger than NSP, but the reverse situation is not immediately clear.
Later Aldroubi et al.~adopted the same approach in \cite{lqnsp}, and proved that NSP and NSP' are in fact equivalent (see also \cite{lqharmonic}). The proof method in \cite{lqnsp} requires a lemma from matrix analysis \cite[Lemma 2.1]{lqnsp}. We remark that this lemma, from a slightly more general viewpoint, can be seen as a classical application of the open mapping theorem in functional analysis \cite[Chapter 4, Corollary 3.2]{stein_func}. Thus it is established that NSP, NSP', ERC and RRC are all equivalent for $\ell_p$-minimization.

While the NSP' approach is nice for the $\ell_p$ case, it is hard to be extended to the general $F$-minimization problem. This is because NSP' consists of a homogeneous inequality, which appears to work well only for homogeneous cost functions such as the $\ell_p$ norm. In contrast, the heart of our approach is the interior point characterization of RRC (Theorem \ref{th2}) for the general $F$-minimization problem. Then our proof of the ERC/RRC equivalence for $\ell_p$-minimization, although involves some basic facts about topological spaces, follows almost immediately as a corollary. Note this application is particularly interesting since the statement of ERC/RRC equivalence does not involve topology at all. Nevertheless, we emphasize that the significance of Theorem \ref{th2} is to provide a simple, accurate, and general characterization of the robustness of $F$-minimization; and the proof of ERC/RRC equivalence for $\ell_p$ is one of its applications in a special setting.

\subsection{The Notion of Sparseness Measure}\label{s5b}
The sparseness measure defines the class of cost functions of our interest, and is therefore of great importance. In general we want to consider a class wide enough to cover most applications,
but also small enough to possess important recovery properties. Intuitively, the cost function should penalize non-zero coefficients, and not penalize the zero coefficients. However there are additional reasonable requirements, the precise definitions of which differ in the literature. For clarifications we compare these different requirements on $F$ as follows (Recall that $\mathcal{M}$ denotes the set of sparseness measures defined in Definition \ref{def1}):

$\bullet$ $F\in\mathcal{M}$. This is the class of functions mainly considered in our paper as well as \cite{lqnsp}. This seems to be most general class of functions that can be studied by the null space property.

$\bullet$ $F\in\mathcal{M}$ and $F$ is non-decreasing. This requirement appears in Theorem \ref{probeq}. As shown in the counter example in the remark following the theorem, the assumption that $F$ being non-decreasing cannot be dropped.

$\bullet$ $F\in\mathcal{M}$, $F$ is non-decreasing, and $F(t)/t$ is non-increasing\footnote{The assumption of $F(t)/t$ being non-increasing guarantees that $F$ is subadditive, as shown in Proposition \ref{prop1}.}. This requirement is considered in \cite{Gribonval,troppdis}, and it guarantees that the cost function $J_F$ is better than $\ell_1$ norm in the sense of ERC.
There is also another nice property relating to the composition of two functions in this class \cite[Lemma 7]{Gribonval}. Finally, $\ell_1$ norm is the only convex cost function whose corresponding $F$ satisfies this definition of sparseness measure \cite[Proposition 2.1]{troppdis}.

\subsection{About the Robustness Constant}
The robustness constant derived in Corollary~\ref{as} is upper bounded asymptotically in the linear scaling regime.
On the other hand, the robustness constant obtained in \cite{lqnsp} is (in the notation of our paper) $n^{1/p-1/2}(\frac{4}{1-\theta_{\ell_p}})^{1/p}\sqrt{\frac{2}{\sigma^2_{\min}(\mathbf{A}^{\top})}}$, which blows up in the linear scaling setting in Corollary \ref{as} as $n\to\infty$.

In the special case of $\ell_1$-minimization, our problem setting and the notion of robustness is also the same as the classical paper \cite{candes2006stable} by Cand\`{e}s. In Theorem 1 of that paper, it is shown that (in the notation of our paper) if the RIP constants $\delta_{3k}+3\delta_{4k}<2$ then RRC is satisfied while robustness constant $C$ may depend on $\delta_{4k}$.
Later, the same author provided a similar but improved result on robustness in \cite{candes2008restricted}, where the assumption depends on $\delta_{2k}$ instead of $\delta_{4k}$.
However the known estimates of RIP constant usually contains implicit constants that are hard to compute.
Moreover, according to a comparative study of \cite{blanchard2011compressed}, performance estimates for exact recovery based on RIP is often not as sharp as analysis based on Gordon's theorem in the proportional growth setting. This implies that Gordon's theorem also provides a better estimate for the robustness constant, since the threshold for exact recovery coincides with the threshold for $C<\infty$ in the $\ell_1$ case.

While the setup of Corollary~\ref{as} is well suited and common for signal processing and communication applications, there are other notions of robustness in other settings. For example, in statistical learning one is often interested in the minimax rates for recoveries of sparse vectors. In \cite{raskutti2011minimax} the $\ell_2$ minimax rates for high-dimensional linear regression over $\ell_q$-balls were derived, and the analysis therein used a similar restricted eigenvalue assumption \cite{raskutti2010restricted}. However the problem studied in \cite{raskutti2011minimax} is different from ours in several notable ways. For instance, \cite{raskutti2011minimax} concerns the \emph{optimal} estimator for a signal belonging to an \emph{$\ell_q$-ball}, whereas our paper considers reconstructing a strictly sparse signal (i.e.~in a certain \emph{$\ell_0$-ball}) via \emph{$\ell_p$-minimization}. For strictly sparse signals and optimal estimators under Gaussian noises, the noise sensitivity $C<\infty$ whenever the number measurement is larger than the sparsity; see \cite{wu2012optimal}.

\section{Conclusion}\label{conclusion}
$F$-minimization refers to a broad family of non-convex optimizations for sparse recovery which has outperformed conventional $\ell_1$ minimization experimentally. However because of some technical difficulties, the robustness of $F$-minimization was not fully understood before, even though its exact recovery property has been studied by using the null space property. The novel approach of this paper is to view the collection of null spaces as a topological manifold, called the Grassmann manifold,
and provide an exact characterization of the relationship between robust recovery condition (RRC) and exact recovery condition (ERC):
the set of null spaces of measurement matrix $\mathbf{A}$ satisfying RRC is the interior of the one satisfying ERC. 
Building on this characterization,
the previous result of the equivalence of exact recovery and robust recovery in the $\ell_p$-minimization follows as an easy consequence. Although the RRC set is in general a proper subset of the ERC set, the difference is only a set of measure zero and of the first category, provided that $F$ satisfies the mild condition of being non-decreasing.
The practical significance of this result is that ERC and RRC will occur with equal probability when the measurement matrix is randomly generated according to a continuous distribution. On the quantitative side, a desired level of robustness can be guaranted if the null space of $\mathbf{A}$ is drawn from the ``$d$-interior'' of $\Omega_J$ for a certain $d$.
Specifically, the null spaces in $d\mbox{-}\interior(\Omega_J)$ satisfies RRC with $C=\frac{2+2d}{d\sigma_{\min}(\mathbf{A}^{\top})}$; and null spaces outside of $d\mbox{-}\interior(\Omega_J)$ cannot satisfy RRC with $C=\frac{2-2d}{d\sigma_{\max}(\mathbf{A}^{\top})}$.

Although our main contribution of clearing up the relation between ERC and RRC is of conceptual nature, its ramifications provide several guidelines for the engineering design:
\begin{enumerate}
    \item Achievability results of the tradeoff between sampling rate and robustness in \ref{sgordon} appear to be tight under certain conditions, which may help engineers to evaluate how well the existing algorithms perform. For example, if $\sigma_{\max}=\sigma_{\min}$ (which can always achieved by taking $\mb{Q}^{\top}$ in the $QR$ factorization of $\mb{A}^{\top}$ as the measurement matrix), and $d$ is small, then the upper and lower bounds on $C$ in Theorem~\ref{th22} are tight. Moreover, it can be shown that the bound in \eqref{48} is tight for small $\frac{\gamma}{\beta}$, and recent experimental and analytical results (see \cite{amelunxen2014living} and the references therein) indicate that estimates on $\mu(\Omega_{\ell_1})$ via Gordon's theorem appear to be tight asymptotically. These suggest that our estimates on $\mu(d\mbox{-}\Omega_J)$, and hence the estimate of $C$ in Corollary~\ref{as1} are asymptotically tight under the above conditions.
    \item In order to have decent worst case performances, $F$ must converge to zero sufficiently fast at the origin, according to the comparison rule Theorem~\ref{th1}. From another perspective, it suggests that there is not much gain to use $F$-minimization instead of $\ell_p$-minimization where $p$ is the exponent of H\"{o}lder continuity of $F$ at the origin, either in the sense of exact recovery or robustness. This is not true if the measurement matrix has a discrete distribution, but we expect that the discrepancy will generally disappear for large $n$ expect in some artificially created bizarre examples.
\end{enumerate}

Further improvements may include finding more general conditions on $F$ than non-decreasing in order that Theorem~\ref{probeq} still holds. Studies of the robustness under perturbation in the measurement matrix may also be of interest. Also, the approach in \ref{sgordon} can be extended to other sparseness measures, provided that an estimate of the Gaussian width of $\mathbf{K}_J$ is available. In fact, upper bounds of the Gaussian width of $\mathbf{K}_{\ell_p}$ is already available in previous research \cite{stojnic}. Another important but challenging problem is whether the asymptotic performances in Corollary~\ref{as} and Corollary~\ref{as1} predicted by Gordon's theorem are also achievable for other random measurement matrices, such as a matrix with i.i.d.~non-Gaussian entries or the sampled Fourier matrix. According to recent experimental and analytical results such as \cite{donoho2009observed}\cite{bayati2012universality}, various phase transition points in compressed sensing are in some sense insensitive of the measurement matrix ensemble.

\appendices
\section{Proof of Theorem \ref{suf2}}\label{p_suf2}
For the direct part, assume that (\ref{eq16}) is true. Suppose $\mathbf{\hat{x}}$ is a feasible vector with $J(\mathbf{\hat{x}})\le J(\mathbf{\bar{x}})$, and we want to show that $\mathbf{\hat{x}}$ is close to $\mathbf{\bar{x}}$. From the constraint of the optimization we have
\begin{equation}
\|\mathbf{A}(\mathbf{\hat{x}}-\mathbf{\bar{x}})\|\le
\|\mathbf{A}\mathbf{\hat{x}}-\mathbf{y}\|
+\|\mathbf{A}\mathbf{\bar{x}}-\mathbf{y}\|
\le 2\epsilon.
\end{equation}
Define $\mathbf{u}:=\mathbf{\bar{x}}-\mathbf{\hat{x}}$; we find that
\begin{align}
J({\bf u}_T)
&\ge J(\bar{\bf x}_T)-J(\hat{\bf x}_T)  \label{step1}\\
&=  J(\bar{\bf x})-J(\hat{\bf x}_T)    \label{step2}\\
&\ge J(\hat{\bf x})-J(\hat{\bf x}_T)   \label{step3}\\
&=  J(\hat{\bf x}_{T^c})\nonumber\\
&=  J({\bf u}_{T^c})\nonumber
\end{align}
Where (\ref{step1}) is from subadditivity of $F$, (\ref{step2}) is because $\bar{\bf x}$ is supported on $T$, and (\ref{step3}) is from the assumption of $\hat{\bf x}$.
Decompose $\mathbf{u}=\mathbf{z}+\mathbf{n}$, such that $\mathbf{z}\in\mathcal{N}(\mathbf{A})$, $\mathbf{n}\in\mathcal{N}(\mathbf{A})^{\bot}$. The above inequality is in contradiction with (\ref{eq16}), hence from the assumption we must have:
\begin{equation}
\|\mathbf{n}\|\ge d\|\mathbf{z}\|,
\end{equation}
which by triangular inequality implies that $\|{\bf n}\|\ge d(\bf \|u\|-\|n\|)$, or $\|{\bf n}\|\ge\frac{d}{1+d}\|\bf u\|$. Therefore
\begin{align}
2\epsilon
&\ge\|\mathbf{A}(\mathbf{\hat{x}}-\mathbf{\bar{x}})\|\nonumber\\
&=\|\mathbf{A}\mathbf{n}\|\nonumber\\
&\ge \sigma_{\min}\|\mathbf{n}\|\nonumber\\
&\ge \sigma_{\min}\frac{d}{1+d}\|\mathbf{u}\|\nonumber\\
&=\sigma_{\min}\frac{d}{1+d}\|\mathbf{\hat{x}}-\mathbf{\bar{x}}\|
,\nonumber
\end{align}
where $\sigma_{\min}$ is the smallest singular value of $\mathbf{A}^{\top}$. Thus RRC holds with $C=\frac{2(1+d)}{d\sigma_{\min}}$.

Conversely, assuming that
\begin{align}\label{cond4}
&\exists d>0, \mathbf{z}\in \mathcal{N}(\mathbf{A})\setminus\{\mathbf{0}\}, \mathbf{n}:\|\mathbf{n}\|<d\|\mathbf{z}\|, T:|T|\le k \nonumber\\
&\textrm{s.t.~} J(\mathbf{z}_T+\mathbf{n}_T)\ge J(\mathbf{z}_{T^c}+\mathbf{n}_{T^c}),
\end{align}
we will show that RRC with $C=\frac{2(1-d)}{d\sigma_{\max}}$ is impossible. To do this, we will construct $\hat{\mathbf{x}}$, $\bar{\mathbf{x}}$ with $J(\bar{\mathbf{x}})\ge J(\hat{\mathbf{x}})$, and $\mathbf{v},\epsilon$ with $\|\mathbf{v}\|=\epsilon$, $\|\mathbf{A}\hat{\mathbf{x}}-(\mathbf{A}\bar{\mathbf{x}}+\mathbf{v})\|=\epsilon$;
but
\begin{equation}\label{eq3}
\|\hat{\mathbf{x}}-\bar{\mathbf{x}}\|> \frac{2(1-d)\epsilon}{d\|\mathbf{A}\|},
\end{equation}
where $\|\mathbf{A}\|=\sigma_{\max}$ denotes the operator norm of matrix $\mathbf{A}$.

Now suppose $d,\mathbf{n},\mathbf{z}$ are as in (\ref{cond4}). Define\footnote{For $\mathbf{x}\in \mathbb{R}^{|T|}$, we denote by $\mathbf{x}^T\in \mathbb{R}^n$ the $n$-vector supported on $T$ satisfying $(\mathbf{x}^T)_T=\mathbf{x}$.} $\bf u:=z+n$, $\hat{\mathbf{x}}:=(\mathbf{u}_T)^T$, $\bar{\mathbf{x}}:=-(\mathbf{u}_{T^c})^{T^c}$, $\mathbf{v}:=\mathbf{A}(\hat{\mathbf{x}}-\bar{\mathbf{x}})/2$, $\epsilon:=\|\mathbf{v}\|$. Then feasibility is satisfied since $\|\mathbf{A}\hat{\mathbf{x}}-(\mathbf{A}\bar{\mathbf{x}}+\mathbf{v})\|=\epsilon$. Also
\begin{align}
2\epsilon&=\|\mathbf{A}(\hat{\mathbf{x}}-\bar{\mathbf{x}})\|\nonumber\\
&=\|\mathbf{A}\mathbf{n}\|\nonumber\\
&\le \|\mathbf{A}\|\|\mathbf{n}\|\nonumber\\
&< \|\mathbf{A}\|\frac{d}{1-d}\|\mathbf{u}\|,\nonumber
\end{align}
where the last step is because $\|{\bf u}\|+\|{\bf n}\|\ge\|{\bf z}\|>\frac{1}{d}\|{\bf n}\|$, which implies that $\|{\bf u}\|>(\frac{1}{d}-1)\|{\bf n}\|$. Thus the relation (\ref{eq3}) holds, as desired.

\section{Proof of Theorem \ref{th2} and Theorem \ref{th22}}\label{p_th2}
The proof of the theorems will be based on the following result:
\begin{lemma}\label{lmth2}
Suppose $\nu\in G_{l}(\mathbb{R}^n)$. For all $\mathbf{z}\in \nu\setminus \{\mathbf{0}\},\|\mathbf{n}\|<d\|\mathbf{z}\|$, there exists $\nu'\in G_{l}(\mathbb{R}^n)$ such that $\mathbf{z}+\mathbf{n}\in \nu'$ and $\dist(\nu,\nu')<d$.
\end{lemma}
\begin{proof}
If $d>1$, then any $\nu'\in G_l(\mathbb{R}^n)$ will satisfy $\dist(\nu,\nu')\le1<d$ as desired. Now suppose $d\le 1$, so that $\bf z+n\neq 0$. Let $\nu_0\subseteq \nu$ be the subspace such that $\dim(\nu_0)=l-1$ and ${\bf z}\bot \nu_0$. Define $\nu'=\spn({\bf z+n})\oplus \nu_0\in G_l(\mathbb{R}^n)$.~\footnote{Here $\oplus$ denotes the direct sum of linear subspaces.} Let $\theta_i, i=1,\dots, l$ be the principal angles between $\nu$ and $\nu'$ in the ascending order. Then by the construction we have
\begin{align}
\theta_i&=0,\quad \forall i=1,\dots,l-1;\\
\theta_l&=\angle({\bf z},{\bf P}_{\nu_0^{\bot}}(\bf z+n))=\angle({\bf z},{\bf z}+{\bf P}_{\nu_0^{\bot}}{\bf n}).
\end{align}
where ${\bf P}_{\nu_0^{\bot}}$ denotes the projection matrix onto the orthogonal complement of $\nu_0$. Then,
\begin{align}
\dist(\nu,\nu')&=\|\bf P_{\nu}-P_{\nu'}\| \\
&=\sin(\theta_l) \label{l61}\\
&=\sin(\angle({\bf z},{\bf z}+{\bf P}_{\nu_0^{\bot}}{\bf n}))\\
&\le \|{\bf P}_{\nu_0^{\bot}}{\bf n}\|/\|\bf z\|\label{l62}\\
&\le \|{\bf n}\|/\|\bf z\|\\
&\le d,
\end{align}
where (\ref{l61}) is from a basic property of the principal angles, see for example \cite{optheory}, and (\ref{l62}) is from elementary geometry. The lemma is proved.
\end{proof}
\begin{proof}[Proof of Theorem \ref{th22}]
If $\nu\in d\mbox{-}\interior(\Omega_J)$, then by definition we have
\begin{equation}\label{leth21}
\nu'\in \Omega_J,\quad\forall d:\dist(\nu,\nu')<d.
\end{equation}
Now for any $\mathbf{z}\in \nu\setminus \{\mathbf{0}\}$, and $\mathbf{n}$ satisfying $\|\mathbf{n}\|<d\|\mathbf{z}\|$, there exist $\nu'$ such that $\mathbf{z}+\mathbf{n}\in \nu'$ and $\dist(\nu,\nu')<d$ by Lemma \ref{lmth2}. Define $\bf z' = z + n\in \nu'$. Since $\nu'\in\Omega_J$ by (\ref{leth21}), we have $J({\bf z}'_T ) < J({\bf z}'_{T^c})$ for all $T$ such that $|T|\le k$, which is exactly (\ref{eq16}).

Conversely, suppose that RRC is satisfied for some $\bf A$ with $C=\frac{2(1-d)}{d\sigma_{\max}}$ for some $d>0$. Let $\nu:=\mathcal{N}(\bf A)$. For an arbitrary $\nu'$ such that $\dist(\nu,\nu')< d$ and $\mb{z}'\in\nu'$, define the projection
\begin{align}
\mb{w}:=\mb{P}_{\nu}\mb{z},
\end{align}
then from the definition of the metric \eqref{def_dis}, the angle $\theta:=\angle(\mb{w},\mb{z}')$ satisfies
\begin{align}
\sin\theta< d.
\end{align}
Now define
\begin{align}
\mb{z}:=\frac{1}{\cos^2\theta}\mb{w}.
\end{align}
It is clear that $\mb{P}_{\nu'}\mb{z}=\mb{z}'$, so
\begin{align}
\frac{\|\mb{z}-\mb{z}'\|}{\|\mb{z}\|}=\sin\theta< d.
\end{align}
Since $\mathbf{z}\in \mathcal{N}(\bf A)\setminus\{\bf 0\}$ and $({\bf A},k,J)$ satisfies RRC with $C=\frac{2(1-d)}{d\sigma_{\max}}$, setting $\bf n=z'-z$ in Theorem \ref{suf2} shows that $J(\mathbf{z}'_T)<J(\mathbf{z}'_{T^c})$ for every $|T|\le k$. Hence $\nu'\subseteq \Omega_J$ from the arbitrariness of $\mathbf{z}'$, and $\nu\subseteq d\mbox{-}\interior(\Omega_J)$ by \eqref{e32}, as desired.
\end{proof}

\begin{proof}[Proof of Theorem \ref{th2}]
Since $\interior(\Omega_J)=\bigcup_{d>0}d\mbox{-}\interior(\Omega_J)$, Theorem \ref{th2} follows directly from Theorem~\ref{th22}.
\end{proof}

\section{Proof of Corollary \ref{th5}}\label{p_cont}
We first note the following basic fact about generic continuous functions. (It is stated in a slightly stronger and more complete manner than needed for proving Corollary \ref{th5}).
\begin{lemma}\label{cont}
Suppose $\mathcal{X},\mathcal{M}$ are metric spaces, and $\overline{\mathbb{R}}=\mathbb{R}\cup\{+\infty,-\infty\}$ be the extended real line. If $f:~\mathcal{X}\times \mathcal{M}\to\overline{\mathbb{R}}$ is continuous, then $g:~\mathcal{X}\to\overline{\mathbb{R}},~x\mapsto\sup_{y\in \mathcal{M}}f(x,y)$ is lower semicontinuous on $\mathcal{X}$. Further, if $\mathcal{M}$ is compact, then $g$ is also continuous.
\end{lemma}
\begin{proof}
The lower semi-continuity of $g$ follows from the fact that $g$ is defined as the supremum of a collection of continuous functions \cite[P38 (c)]{rudin}. 
To show that $g$ is also upper semicontinuous when $\mathcal{M}$ is compact, we will prove that $g$ is upper semicontinuous at an arbitrary $x_0\in \mathcal{X}$: let $y_0$ be a point in $\mathcal{M}$ such that $g(x_0)=f(x_0,y_0)$ (Here we used the compactness of $\mathcal{M}$). Suppose otherwise, that $g$ is not supper semicontinuous at $x_0$, then there exists $\epsilon> 0$ such that:
\begin{equation}
\limsup_{x\to x_0}g(x) > g(x_0) + \epsilon.
\end{equation}
This implies that we can find sequences $x_n,y_n(n\ge 1)$ such that $\lim_{n\to \infty}x_n=x_0$ and the following holds:
\begin{equation}
f(x_n,y_n) > g(x_0) + \epsilon.
\end{equation}
Since $\mathcal{M}$ is compact, we can find a subsequence $y_{n_k},(k\ge 1)$ converging to some
point $y^*\in \mathcal{M}$. Hence
\begin{align}
g(x_0)&=f(x_0,y_0)\nonumber\\
&\ge f(x_0,y^*)\nonumber\\
&=\lim_{k\to\infty}f(x_{n_k},y_{n_k})\nonumber\\
&\ge g(x_0)+\epsilon,\nonumber
\end{align}
which is an apparent contradiction.
\end{proof}
\begin{remark}
In the above proof, the assumption that $\mathcal{X},\mathcal{M}$ are metrical spaces rather than topological spaces is useful only in showing the existence of the sequences $x_n,y_n,(n\ge1)$. Therefore, the result actually holds when $\mathcal{X},\mathcal{M}$ are topological spaces satisfying the first countable theorem \cite{munk}.
\end{remark}

It then follows the following result about the null space constant $\theta_J$, now conceived as a map from $G_l(\mathbb{R}^n)$ to the real numbers:
\begin{lemma}\label{le3}
If $F$ is continuous, then $\theta_{J}:G_{l}(\mathbb{R}^n)\to [0,+\infty)$ is a lower semicontinuous function. Further, $\theta_{\ell_p}:G_{l}(\mathbb{R}^n)\to [0,+\infty)$ is a continuous function.
\end{lemma}
\begin{proof}
It suffices to show that $\theta_{J}$ is lower semicontinuous or continuous on each $U_I$. Without loss of generality, we may assume that $I=\{1,\dots,l\}$. For generic $F$, let $\mathcal{X}=\mathbb{M}(n-l,l)$, $\mathcal{M}=\mathbb{R}^l\setminus\{\bf 0\}$, and $f: \mathcal{X}\times\mathcal{M}\to\overline{\mathbb{R}},(\mathbf{X},\mathbf{y})
\mapsto\frac{J(\mathbf{z}_T)}{J(\mathbf{z}_{T^c})}$, where $\bf z:=\left(
                                                                     \begin{array}{c}
                                                                       I \\
                                                                       X \\
                                                                     \end{array}
                                                                   \right)y
$. Then $\theta_J(\phi_I^{-1}(\mathbf{X}))=\sup_{y\in\mathcal{M}}f(\bf X,y)$, which by Lemma \ref{cont} implies that the composition map $\theta_J\circ \phi_I^{-1}$ is lower semicontinuous. Since $\phi_I$ is a homeomorphism, we conclude that $\theta_J$ is also lower semicontinuous.

For the case of $\ell_p$-minimization, we can define $\mathcal{M}:=S^{l-1}$, while $\mathcal{X}$ and $f$ are as before. By homogeneity we still obtain $\theta_J(\phi_I^{-1}(\mathbf{X}))=\sup_{y\in\mathcal{M}}f(\bf X,y)$. But since $\mathcal{M}$ is compact in this case, we conclude that $\theta_{\ell_p}$ is continuous.
\end{proof}

The openness of $\Omega_{\ell_p}$ then follows easily, from the very definition of continuous functions: that the pre-images of open sets are open.

\begin{proof}[Proof of Corollary \ref{th5}]
By Lemma \ref{le3}, function $\theta_{\ell_p}$ is continuous with respect to $\nu$. Since $\Omega_{\ell_p}$ is the pre-image of $(-\infty,1)$ under the continuous mapping of $\theta_{\ell_p}$ (Lemma \ref{le1}), we conclude that $\Omega_{\ell_p}$ is open, hence $\Omega^r_{\ell_p}=\interior(\Omega_{\ell_p})=\Omega_{\ell_p}$.
\end{proof}

\section{Proof of assertions in Counter-example \ref{ex1}}\label{appex}
For any $\mathbf{w}\in\mathcal{N}$ we can write $\mathbf{w}=(xt,yt,zt)^\top$ for some $t\in\mathbb{R}$. Since $|zt|>|xt|,|yt|$, by strict subadditivity $F(xt)+F(yt)>F(zt)$ holds. Then for any $T$ such that $|T|=1$ we have:
\begin{equation}
J(\mathbf{w}_T)<J(\mathbf{w}_{T^c}).
\end{equation}
Hence NSP is satisfied, and ERC must hold.
On the other hand, the above inequality fails under arbitrarily small perturbation: for any $0<d<1$, Taylor expansion yields $F((1-d)xt)+F(yt)=2(1-d)xt+2yt+o(t^2)=2zt-2dxt+o(t^2)$ and $F(zt)=2zt+o(t^2)$ (for small $t$), so there exist $t>0$ such that
\begin{equation}\label{eq5}
F((1-d)xt)+F(yt)<F(zt).
\end{equation}
Now in Theorem \ref{suf2}, take $\mathbf{z}=(xt,yt,zt)^\top$, $T=\{3\}$, and $\mathbf{n}=(-dxt,0,0)$. On the one hand we have $\|\mathbf{n}\|/\|\mathbf{z}\|\le d$; on the other hand (\ref{eq16}) doesn't hold because of (\ref{eq5}). Therefore RRC is not fulfilled as a result of Theorem \ref{suf2}.

\section{Proof of Theorem \ref{probeq}}\label{p_probeq}
%
%

\begin{definition}{\cite{stein}}
Suppose $E$ is a measurable set in $\mathbb{R}^L$, the \emph{Lebesgue density} of $E$ at a point $\mathbf{x}\in \mathbb{R}^L$ is defined as $
\lim_{r\to 0}\frac
{\lambda(B(\mathbf{x},r)\cap E)}
{\lambda(B(\mathbf{x},r))}
$ where $\lambda$ denotes the Lebesgue measure. If the density exists and is equal to $1$, $\mathbf{x}$ is said to \textit{have the Lebesgue density of $E$}.
\end{definition}

The following result can be found in standard textbook on measure theory, such as \cite[Chapter 3, Corollary 1.5]{stein}.
\begin{theorem}[Lebesgue density theorem]
If $E$ is a measurable set in $\mathbb{R}^L$, then almost all $\mathbf{x}\in E$ (except for a set of Lebesgue measure zero) has the Lebesgue density of $E$.
\end{theorem}

We say a set $E\subseteq\mathbb{R}^L$ is \emph{porous} at a point $\mathbf{x}\in \mathbb{R}^L$, if there exists $r_0>0$, $0<\alpha< 1$ such that for each $0<r<r_0$, there exists some $\mathbf{y} \in \mathbb{R}^L$ such that the ball $B(\mathbf{y},\alpha r)\subseteq B(\mathbf{x},r)\setminus E$. If $E$ is porous at each point in $E$, then $E$ must be of measure zero, which is a direct consequence of Lebesgue density theorem; and $E$ must be of the first category from definition.

As in many other problems from analysis, a trivial observation is that the ``ball of radius $a$'' (where $a=r$ or $\alpha r$) in the definition of porosity can be replaced with, say, a ``hypercube of edge length $a$'' (with $\alpha$ taking a possibly different value), since any hypercube of edge length $a$ is contained in a ball of radius $\lambda_{\max}a$, and contains a ball of radius $\lambda_{\min}a$, where $0<\lambda_{\min}<\lambda_{\max}$ are constants independent of $a$. Another observation based on the same mechanism is that porosity is preserved under invertible linear transforms, since the image of a unit ball under the linear transform must contain a ball whose radius is the least singular value of the linear transform. Further, since a smooth function can be locally approximated by a linear map by Taylor expansions, we see that any $C^{\infty}$ function which has a $C^{\infty}$ inverse also preserves porosity.

\begin{proof}[Proof of Theorem \ref{probeq}]
We shall adopt the notation in (\ref{ot}). For any $\nu\in S-\interior(\Omega_T)$, there exists a sequence $\nu^l\in \Omega^c_T, l=1,2,\dots,\infty$. With the properties:
\begin{flushleft}
1) $\nu^l\to\nu$ as $l\to\infty$. \\
2) For each $l$ there exist $\mathbf{z}^l\in \nu^l-\{0\}$ such that $J(\mathbf{z}^l_T)\ge J(\mathbf{z}^l_{T^c})$.\\
3) The sequence $\mathbf{\bar{z}}^l:=\mathbf{z}/\|\bf z\|$ converges to some $\mathbf{x}\in S^{n-1}$.
\end{flushleft}
Property 1) is from the fact that $\nu$ is in the closure of $\Omega^c_T$. Property 2) is from the definition of $\Omega_T$. As for property 3), we note that the compactness of $S^{n-1}$ implies there exists a convergent subsequence of $\mathbf{\bar{z}}^l$. So if the $\mathbf{\bar{z}}^l$ sequence itself is not convergent, we can redefine $\mathbf{\bar{z}}^l$ to be its convergent subsequence, and then redefine the sequences $\nu^l$, $\mathbf{z}^l$ to be their corresponding subsequences. As a result, Properties 1), 2), 3) can always be satisfied. Notice that
\begin{align}
\|\mathbf{P}_{\nu}-\mathbf{x}\|&\le \|\bar{\mathbf{z}}^l-\mathbf{x}\|+\|\mathbf{P}_{\nu}\mathbf{x}-\bar{\mathbf{z}}^l\|\\
&\le \|\bar{\mathbf{z}}^l-\mathbf{x}\|+\|\mathbf{P}_{\nu}\bar{\mathbf{z}}^l-\bar{\mathbf{z}}^l\|\\
&= \|\bar{\mathbf{z}}^l-\mathbf{x}\|+\|\mathbf{P}_{\nu}\bar{\mathbf{z}}^l-\mathbf{P}_{\nu^l}\bar{\mathbf{z}}^l\|\\
&\le \|\bar{\mathbf{z}}^l-\mathbf{x}\|+\|\mathbf{P}_{\nu}-\mathbf{P}_{\nu^l}\|\\
&\to 0,\quad\textrm{as~}l\to\infty,
\end{align}
therefore $\|\mathbf{P}_{\nu}-\mathbf{x}\|=0$ and $\mathbf{x}\in \nu$.

Suppose $\mathbf{x}=(x_1,\dots,x_n)^\top$. Since $\nu\in S$, at most $l-1$ of the entries of $\mathbf{x}$ can be zero. Hence there exists an $l$-element index set $I_0\subseteq\{1,\dots,n\}$ such that $\{x_i~|~i\in I_0\}$ has at least one non-zero element and $\{x_i~|~i\notin I_0\}$ has no zero element. Without loss of generality, let us assume that $I_0=\{1,\dots,l\}$ and $x_1\neq 0$. Consider the chart $({\bf B}\circ \phi_{I_0},U_{I_0})$, where ${\bf B}\circ\phi_{I_0}$ is the composite of $\phi_{I_0}:U_{I_0}\to\mathbb{M}(n-l,l)$, and the right multiplication of matrix $$\mathbf{B}:=
\left(
  \begin{array}{c|c}
    x_1 & 0 \\
    \hline
    x_2 & ~ \\
    \vdots & {\bf I}_{(l-1)\times(l-1)} \\
    x_l & ~ \\
  \end{array}
\right).
$$
Then for each $\nu\in G_l(\mathbb{R}^n)$, it holds that
\begin{equation}
\pi(\left(
      \begin{array}{c}
        \mathbf{B} \\
        \hline
        \mathbf{B}\circ\phi_{I_0}(\nu) \\
      \end{array}
    \right)
)
=\pi(\left(
       \begin{array}{c}
         \mathbf{I} \\
         \hline
         \phi_{I_0}(\nu) \\
       \end{array}
     \right)
)
=\nu,
\end{equation}
where we recall that $\pi$ is the projection map to the subspace spanned by the column vectors of a matrix. Notice that the first column of matrix $\left(
      \begin{array}{c}
        \mathbf{B} \\
        \hline
        \mathbf{B}\circ\phi_{I_0}(\nu) \\
      \end{array}
    \right)$ is $\mathbf{x}$, because its first $l$ elements of this column agree with $\mathbf{x}$ by definition of $\mathbf{B}$, and then the rest of the $n-l$ elements must also agree with $\mathbf{x}$ since $\mathbf{x}$ is a unique linear combination of columns of $\left(
      \begin{array}{c}
        \mathbf{B} \\
        \hline
        \mathbf{B}\circ\phi_{I_0}(\nu) \\
      \end{array}
    \right)$. Now define
\begin{align}\label{pv}
V:=\{\mathbf{M}\in\mathbb{M}(n-l,l)|~|M_{i1}|>|x_{i+l}|~\textrm{if}~i+l\in T {\rm~and~}
|M_{i1}|<|x_i|~\textrm{otherwise, where $1\le i\le n-l$}\}.
\end{align}
Since by assumption $|x_{i+l}|>0$ for $i=1,\dots,n-l$, the set $V$ is not empty.

Next we shall show that for each $\mathbf{M}\in V$, we have
\begin{align}\label{star}
\pi(\left(
      \begin{array}{c}
        \mathbf{B} \\
        \mathbf{M} \\
      \end{array}
    \right)
)
\notin \Omega_T,
\end{align}
which, by setting $\nu:=(\mathbf{B}\circ\phi_{I_0})^{-1}(\mathbf{M})$, will imply that $M\in \mathbf{B}\circ\phi_{I_0}(U_{I_0}\setminus\Omega_T)$, or $V\subseteq \mathbf{B}\circ\phi_{I_0}(U_{I_0}\setminus\Omega_T)$. Since $V$ is open, this in turn implies that
\begin{align}\label{89}
V\subseteq \interior(\mathbf{B}\circ\phi_{I_0}(U_{I_0}\setminus\Omega_T))
=\mathbf{B}\circ\phi_{I_0}(\interior(U_{I_0}\setminus\Omega_T))
=\mathbf{B}\circ\phi_{I_0}(U_{I_0}\setminus\overline{\Omega}_T).
\end{align}
To see (\ref{star}), consider a vector $\mathbf{c}_{\epsilon_1,\epsilon_2}:=(1+\epsilon_1s_1,0,\dots,0)^{\top}+\epsilon_2(0,s_2,\dots,s_l)^{\top}\in\mathbb{R}^l$, where
\begin{align}
s_1:=\left\{
\begin{array}{cc}
  1 & \textrm{if }i\in T; \\
  -1 & \textrm{otherwise},
\end{array}
\right.
\end{align}
and for $1<i\le l$,
\begin{align}
s_i:=
\left\{
\begin{array}{cc}
  \sign(x_i) & \textrm{if }x_i\neq0, 1\in T; \\
  -\sign(x_i) & \textrm{if }x_i\neq0, i\in T^c; \\
  0 & \textrm{if }x_i=0, i\in T^c; \\
  1 & \textrm{if }x_i=0, i\in T.
\end{array}
  \right.
\end{align}
Then consider the vector $\mathbf{v}:=\left(
               \begin{array}{c}
                 \mathbf{B} \\
                 \hline
                 \mathbf{M} \\
               \end{array}
             \right)
             \mathbf{c}_{\epsilon_1,\epsilon_2}
             \in \pi(\left(
      \begin{array}{c}
        \mathbf{B} \\
        \mathbf{M} \\
      \end{array}
    \right)
)
$. For fixed $\mathbf{M}\in V$, we wish to show that there exist $0<\epsilon_1<<\epsilon_2<<1$ such that the components of $\mathbf{v}$ satisfies
\begin{flushleft}
a) $|v_i|>|x_i|$, if $i\in T$;\\
b) $|v_i|<|x_i|$, if $i\notin T$ and $|x_i|\neq 0$;\\
c) $v_i=0$, if $i\notin T$ and $|x_i|= 0$.\\
\end{flushleft}
For $l<i\le n$, by the construction of (\ref{pv}), properties a) b) c) are obviously true for small $\epsilon_1,\epsilon_2$. As for $1\le i\le l$, if $i\in T$, we compute
\begin{align}
v_1&=(1+\epsilon_1)x_1;\\
v_i&=\left\{\begin{array}{cc}
             (1+\epsilon_1s_1)x_i+\sign(x_i)\epsilon_2 & \textrm{ if  } x_i\neq0; \\
             \epsilon_2 &  \textrm{if } x_i=0;
           \end{array}
           \right.
           \textrm{for } 1<i\le l,
\end{align}
and similarly if $i\notin T$,
\begin{align}
v_1&=(1-\epsilon_1)x_1;\\
v_i&=\left\{\begin{array}{cc}
             (1+\epsilon_1s_1)x_i-\sign(x_i)\epsilon_2 & \textrm{ if  } x_i\neq0; \\
             0 &  \textrm{if } x_i=0;
           \end{array}
           \right.
           \textrm{for } 1<i\le l.
\end{align}
Thus a), b), c) are guaranteed if $0<\epsilon_1<<\epsilon_2<<1$. Since $\lim_{l\to\infty}\mathbf{\bar{z}^l}=\mathbf{x}$, by a), b), c) there exists $l$ such that
\begin{align}
|v_i|\ge|\bar{z}^l_i|, \textrm{ if }i\in T;\\
|v_i|\le|\bar{z}^l_i|, \textrm{ if }i\notin T;\\
\end{align}
Therefore from the monotonicity of $F$, we get
\begin{align}
J(\|\mathbf{z}^l\|\mathbf{v}_T)&=\sum_{i\in T}F(\|\mathbf{z}^l\|v_i)\\
&\ge\sum_{i\in T}F(\|\mathbf{z}^l\|\bar{z}^l_i)\\
&=\sum_{i\in T}F(z^l_i)\\
&\ge\sum_{i\in T^c}F(z^l_i)\\
&=\sum_{i\in T^c}F(\|\mathbf{z}^l\|\bar{z}^l_i)\\
&\ge \sum_{i\in T^c}F(\|\mathbf{z}^l\|v_i)\\
&=J(\|\mathbf{z}^l\|\mathbf{v}_{T^c}).
\end{align}
Since $\|\mathbf{z}^l\|\mathbf{v}_T\in\pi(\left(
                                            \begin{array}{c}
                                              \mathbf{B} \\
                                              \mathbf{M} \\
                                            \end{array}
                                          \right)
)$, (\ref{star}) is proved.

Finally, from (\ref{pv}) we see that for each hypercube centered at $\mathbf{B}\circ \phi_{I_0}(\nu)=(x_{l+1},\dots,x_n)^{\top}$ with edge length less than $\min_{l<i\le n}\{|x_i|\}$, it must contain a hypercube of half of it's edge length in $V$, which, by (\ref{89}), does not intersect with $\mathbf{B}\circ\phi_{I_0}(S\cap(\overline{\Omega}_T\setminus\interior(\Omega_T)))$. This shows that $\mathbf{B}\circ\phi_{I_0}(S\cap(\overline{\Omega}_T\setminus\interior(\Omega_T)))$ is porous at $\mathbf{B}\circ \phi_{I_0}(\nu)$. But the map $\phi_I\circ\phi_{I_0}^{-1}\circ\mathbf{B}^{-1}:\mathbf{B}\circ\phi_{I_0}(S)\to\phi_{I}(S)$ and its inverse are $C^{\infty}$, therefore by the previous discussion, the set $\phi_{I}(S\cap(\overline{\Omega}_T\setminus\interior(\Omega_T))$ is also porous at $\phi_{I}(\nu)$. Recalling that $\nu$ is arbitrarily chosen from $S\cap(\overline{\Omega}_T\setminus\interior(\Omega_T))$, we see $\phi_I(S\cap(\overline{\Omega}_T\setminus\interior(\Omega_T)))$ must be of measure zero and of the first category. Then by the observation made in (\ref{27}), the proof is complete.
\end{proof}

\section{Proof of Corollary \ref{co2}}\label{app_F}
Suppose $\mu(\Omega_J\setminus\Omega_J^r)=0$. Let $H\subseteq G_m(\mathbb{R}^n)$ be the set of orthogonal complements of the $l$-dimensional subspaces in $\Omega_J\setminus\Omega_J^r$. Since $G_m(\mathbb{R}^n)$ is isomorphic to $G_l(\mathbb{R}^n)$ (recall that $l:=n-m$), we have $\mu(H)=0$ as well\footnote{Here we abused the notation by denoting $\mu$ the Haar measure both on $G_l(\mathbb{R}^n)$ and on $G_m(\mathbb{R}^n)$, since the two manifolds are isomorphic.}.
Then for each $U_I$ (here $I\subseteq\{1,\dots,n\}$ and $|I|=m$), $H\cap U_I$ is a measure zero subset of $U_I$. By the property of product measure we have that $\phi_I(H\cap U_I)\times \mathbb{M}(m,m)$ is a measure zero subset of $\mathbb{M}(l,m)\times\mathbb{M}(m,m)$. Without loss of generality, let us assume that $I=\{1,\dots,m\}$. Define a $C^{\infty}$ map:
\begin{align}
f:\mathbb{M}(l,m)\times\mathbb{M}(m,m)&\to \mathbb{R}^{mn};\\
(\mathbf{M},\mathbf{V})&\mapsto \left(
                                  \begin{array}{c}
                                    \mathbf{I} \\
                                    \mathbf{M} \\
                                  \end{array}
                                \right)
                                \mathbf{V}.
\end{align}
Then $\pi^{-1}(U_I\cap H)=f(\phi_I(U_I\cap H),\mathbb{M}(m,m))$ is a measure zero subset of $\pi^{-1}(U_I)$. Finally $\pi^{-1}(H)=\bigcup_{I}\pi^{-1}(H\cap U_I)$ is a measure zero subset of $\pi^{-1}(G_m(\mathbb{R}^n))=\mathbb{M}(n,m)\setminus\{\mathbf{0}\}$. This shows that $\mathcal{N}(\mathbf{A})\in(\Omega_J\setminus\Omega^r_J))$ only if $\mathbf{A}$ falls into a Lebesgue measure zero set on $\mathcal{M}(m,n)$, which is of probability zero if the probability distribution of $\mathbf{A}$ is absolutely continuous (with respect to the Lebesgue measure).

\section{Proof of Theorem \ref{th1}}\label{ap2}
Since the value of $\theta_J$ depends on the null space of the measurement matrix, in the following it is considered as a function of $\nu\in G_{l}(\mathbb{R}^n)$.
\begin{lemma}
Let $F\in\mathcal{M},q\in(0,1]$. If $\lim_{x\downarrow 0}F(x)/x^q$ or $\lim_{x\to\infty}F(x)/x^q$ exists and is positive, then $\theta_{\ell_q}\le\theta_J$ for any $\nu\in G_{l}(\mathbb{R}^n)$.
\end{lemma}
\begin{proof}
We only prove for the case where $\lim_{t\downarrow 0}F(t)/t^q$ exists and is positive, because the case where $\lim_{t\to \infty}F(t)/t^q$ exists and is positive is essentially similar. By definition we only have to prove the following for any $\mathbf{z}\in \nu\setminus \{\textbf{0}\}$ and $T$ satisfying $|T|\le k$:
\begin{equation}\label{tosee}
\frac{\|\mathbf{z}_{T}\|^q_q}{\|\mathbf{z}_{T^c}\|^q_q}\le \theta_{J}.
\end{equation}
Notice that for any $t\in \mathbb{R}$, vector $t\mathbf{z}$ still belongs to $\mathcal{N}(\mathbf{A})$, hence
\begin{align}
\textrm{left side of}(\ref{tosee})&=\lim_{t\downarrow 0}
\frac{J(t\mathbf{z}_{T})}{J(t\mathbf{z}_{T^c})}\label{lhospital}\\
&\le\theta_{J},\label{supdef}
\end{align}
where (\ref{lhospital}) is because $\lim_{t\downarrow 0}F(t)/t^q$ exists and is positive, and (\ref{supdef}) is from the definition of supremum.
\end{proof}

\begin{lemma}\label{cl}
$\overline{\Omega_{\ell_q}}=\{\nu~|~\theta_{\ell_q}\le 1\}$.
\end{lemma}
\begin{proof}
First we prove that $\overline{\Omega_{\ell_q}}\subseteq\{\nu~|~\theta_{\ell_q}\le 1\}$. This is because Lemma \ref{le3} shows that $\{\nu~|~\theta_{\ell_q}\le 1\}$ is closed. $\Omega_{\ell_q}\subseteq\{\nu~|~\theta_{\ell_q}\le 1\}$.
On the other hand, it is obvious that $\{\nu~|~\theta_{\ell_q}\le 1\}\subseteq\overline{\Omega_{\ell_q}}$. The proof is complete.
\end{proof}
\begin{lemma}
Given $\nu\in G_{l}(\mathbb{R}^n)$, if $\theta_{\ell_q}\le\theta_J$, then $\Omega_J\subseteq \overline{\Omega_{\ell_q}}$.
\end{lemma}
\begin{proof}
By Lemma \ref{nspcond2}, the assumptions imply that
\begin{equation}
\Omega_J\subseteq\{\nu~|~\theta_J\le 1\}\subseteq\{\nu~|~\theta_{\ell_q}\le 1\}=\overline{\Omega_{\ell_q}}.
\end{equation}
\end{proof}
Theorem \ref{th1} then follows easily from the following lemma:
\begin{lemma}
\begin{equation}
\mu(\overline{\Omega_{\ell_q}})
=\mu(\Omega_{\ell_q}).
\end{equation}
\end{lemma}
\begin{proof}
This follows immediately from Theorem \ref{probeq}, with $J$ being the $\ell_p$ norm.
\end{proof}

\section{Proof of Proposition \ref{prop1}}\label{prop1_proof}
\begin{proof}[Proof of (1)]
Suppose $0<t_1<t_2$. From concavity we have
\begin{equation}
F(t_1)\ge \frac{t_2-t_1}{t_2}F(0)+\frac{t_1}{t_2}F(t_2)=\frac{t_1}{t_2}F(t_2).
\end{equation}
Therefore $F(t_1)/t_1\ge F(t_2)/t_2$, which implies that $F(t)/t$ is non-increasing on $(0,+\infty)$.
\end{proof}
\begin{proof}[Proof of (2)]
For arbitrary $0<t_1<t_2$ we have
\begin{align}
\frac{F(t_1)}{t}&=\frac{F(t_1)}{t_1^p}\cdot t_1^{p-1}\\
&\ge \frac{F(t_2)}{t_2^p}\cdot t_1^{p-1}\\
&=\frac{F(t_2)}{t_2}\cdot (\frac{t_2}{t_1})^{1-p}\\
&\ge \frac{F(t_2)}{t_2}
\end{align}
where the inequalities used the fact that $F(t)/t^p$ is non-increasing, and that $1-p\ge 0$. Thus $F(t)/t$ is also non-increasing.
\end{proof}

\begin{proof}[Proof of (3)]
For arbitrary $t_1,t_2>0$, the the assumption that $F(t)/t$ is non-increasing implies that
\begin{align}
F(t_1+t_2)&=\frac{F(t_1+t_2)}{t_1+t_2}\cdot(t_1+t_2)\\
&\ge (\frac{t_1}{t_1+t_2}\frac{F(t_1)}{t_1}+\frac{t_2}{t_1+t_2}\frac{F(t_2)}{t_2})\cdot(t_1+t_2)\\
&=F(t_1)+F(t_2).
\end{align}
Also $F(t_1+t_2)=F(t_1)+F(t_2)$ clearly holds in the case where $t_1=0$ or $t_2=0$. Thus $F$ is subadditive.
\end{proof}

\section{Proof of Lemma~\ref{lem6}}\label{lem6proof}
Fix an arbitrary ${\bf x}\in\mathcal{K}_{\ell_p}(n,k)$. Partition $T^c$ into $L\ge1$ subsets
\begin{align}
T^c=\bigcup_{l=1}^LT_l,
\end{align}
where $|T_l|=b$ for $k=1,\dots,L-1$ and $|T_L|\le b$, such that the entries of ${\bf x}_{T_l}$ have larger absolute values than those of ${\bf x}_{\bigcup_{j=1}^{l-1}T_j}$. We have by triangle inequality
\begin{align}\label{a0}
\|{\bf Ax}\|_2\ge\|{\bf A}_{T_{01}}{\bf x}_{T_{01}}\|_2-\sum_{l=2}^L\|{\bf A}_{T_l}{\bf x}_l\|_2,
\end{align}
where $T_{01}:=T\cup T_1$.
To bound the second term above, consider
\begin{align}\label{a1}
\|{\bf A}_{T_l}{\bf x}_{T_l}\|_2\le \sqrt{\phi_{\max}(b)}\|{\bf x}_{T_l}\|_2,
\end{align}
and
\begin{align}
\|{\bf x}_{T_{l+1}}\|_2&\le
\sqrt{b\max_{i\in T_{l+1}}|x_i|^2}
\\
&\le b\min_{i\in T_l}|x_i|
\\
&\le b\cdot\frac{\|{\bf x}_{T_l}\|_p}{b^{\frac{1}{p}}}
\\
&=\frac{1}{b^{\frac{1}{p}-\frac{1}{2}}}\|{\bf x}_{T_l}\|_p
\end{align}
which yields
\begin{align}
\sum_{l=2}^L\|{\bf x}_{T_l}\|_2
&\le\frac{1}{b^{\frac{1}{p}-\frac{1}{2}}}\sum_{l=2}^L\|{\bf x}_{T_l}\|_p
\\
&\le \frac{1}{b^{\frac{1}{p}-\frac{1}{2}}}\|{\bf x}_{T^c}\|_p \label{rmink}
\\
&\le \frac{\|{\bf x}_{T}\|_p}{b^{\frac{1}{p}-\frac{1}{2}}} \label{ad}
\\
&\le \left(\frac{k}{b}\right)^{\frac{1}{p}-\frac{1}{2}}\|{\bf x}_T\|_2,\label{a2}
\end{align}
where \eqref{rmink} is from the reverse Minkowski inequality and \eqref{ad} is because ${\bf x}\in\mathcal{K}_{\ell_p}(n,k)=\mathcal{D}_{\ell_p}(n,k)$.
Now \eqref{a0}, \eqref{a1} and \eqref{a2} give
\begin{align}
\|{\bf Ax}\|_2
\ge \left(\sqrt{\phi_{\min}(b+k)}-\sqrt{\phi_{\max}(b)}\left(\frac{k}{b}\right)^{\frac{1}{p}-\frac{1}{2}}\right)\|{\bf x}_{T_{01}}\|_2
\end{align}
But from triangle inequality, \eqref{a2} gives
\begin{align}
\|{\bf x}\|_2\le& \|{\bf x}_{T_{01}}\|_2+\|{\bf x}_{T_{01}^c}\|_2
\\
\le&\left(1+\left(\frac{k}{b}\right)^{\frac{1}{p}-\frac{1}{2}}\right)\|{\bf x}_{T_{01}}\|_2,
\end{align}
hence $\|{\bf Ax}\|_2\ge c\|{\bf x}\|_2$, as desired.


\ifCLASSOPTIONcaptionsoff
  \newpage
\fi

\bibliographystyle{ieeetr}
\bibliography{refs1}

\begin{thebibliography}{10}

\bibitem{DC}
E.~J. Cand\`{e}s and T.~Tao, ``Decoding by linear programming,'' {\em IEEE
  Trans. Inf. Theory}, vol.~51, no.~12, pp.~4203--4215, 2005.

\bibitem{Donoho}
D.~L. Donoho, ``Compressed sensing,'' {\em IEEE Trans. Inf. Theory}, vol.~52,
  no.~4, pp.~1289--1306, 2006.

\bibitem{zap}
J.~Jin, Y.~Gu, and S.~Mei, ``A stochastic gradient approach on compressive
  sensing signal reconstruction based on adaptive filtering framework,'' {\em
  IEEE Journal of Selected Topics in Signal Processing}, vol.~4, no.~2,
  pp.~409--420, 2010.

\bibitem{OMP}
J.~A. Tropp and A.~C. Gilbert, ``Signal recovery from random measurements via
  orthogonal matching pursuit,'' {\em IEEE Trans. Inf. Theory}, vol.~53,
  no.~12, pp.~4655--4666, 2007.

\bibitem{donoho2}
S.~Chen, D.~Donoho, and M.~Saunders, ``Atomic decomposition by basis pursuit,''
  {\em SIAM Journal on Scientific Computing}, vol.~20, no.~1, pp.~33--61.

\bibitem{chartrand}
R.~Chartrand and V.~Staneva, ``Restricted isometry properties and nonconvex
  compressive sensing,'' {\em Inverse Problems}, vol.~24, no.~3, pp.~20--35,
  2008.

\bibitem{Foucart1}
S.~Foucart and M.~Lai, ``Sparsest solutions of underdetermined linear systems
  via $\ell_q$-minimization for $0<q\le 1$,'' {\em Applied and Computational
  Harmonic Analysis}, vol.~26, no.~3, pp.~395--407, 2009.

\bibitem{Gribonval}
R.~Gribonval and M.~Nielsen, ``Highly sparse representations from dictionaries
  are unique and independent of the sparseness measure,'' {\em Applied and
  Computational Harmonic Analysis}, vol.~22, no.~3, pp.~335--355, 2007.

\bibitem{sl0}
H.~Mohimani, M.~B. Zadeh, and C.~Jutten, ``A fast approach for overcomplete
  sparse decomposition based on smoothed $\ell_0$ norm,'' {\em IEEE Trans.
  Signal Processing}, vol.~57, no.~1, 2009.

\bibitem{Fan}
J.~Fan and R.~Li, ``Variable selection via nonconcave penalized likelihood and
  its oracle properties,'' {\em Journal of the American Statistical
  Association}, vol.~96, no.~456, pp.~1348--1360, 2001.

\bibitem{MCP}
C.-H. Zhang, ``Nearly unbiased variable selection under minimax concave
  penalty,'' {\em The Annals of Statistics}, vol.~38, no.~2, pp.~894--942,
  2010.

\bibitem{chartrand1}
R.~Chartrand and W.~Yin, ``Iteratively reweighted algorithms for compressive
  sensing,'' {\em ICASSP 2008}, pp.~3869--3872, April.

\bibitem{irls}
I.~Daubechies, R.~DeVore, M.~Fornasier, and C.~S. Gunturk, ``Iteratively
  re-weighted least squares minimization for sparse recovery,'' {\em Princeton
  Univ NJ Program in Applied and Computational Mathematics}, vol.~63, no.~1,
  p.~35.

\bibitem{it}
I.~Daubechies, M.~Defrise, and C.~DeMol, ``An iterative thresholding algorithm
  for linear inverse problems,'' {\em Princeton Univ NJ Program in Applied and
  Computational Mathematics}, vol.~63, no.~1, p.~35.

\bibitem{wang}
X.~Wang, Y.~Gu, and L.~Chen, ``Proof of convergence and performance analysis
  for sparse recovery via zero-point attracting projection,'' {\em IEEE Trans.
  Signal Processing}, vol.~60, no.~8, pp.~4081--4093, 2012.

\bibitem{laming}
L.~Chen and Y.~Gu, ``Approximate projected generalized gradient methods with
  sparsity-inducing penalties,'' {\em Manuscript Submitted to IEEE Trans.
  Signal Processing}.

\bibitem{lqnsp}
A.~Aldroubi, X.~Chen, and A.~Powell, ``Stability and robustness of $\ell_q$
  minimization using null space property,'' {\em Proceedings of SampTA 2011},
  2011.

\bibitem{davenport2011introduction}
M.~A. Davenport, M.~F. Duarte, Y.~C. Eldar, and G.~Kutyniok, ``Introduction to
  compressed sensing,'' {\em Preprint}, vol.~93, 2011.

\bibitem{Foucart}
S.~Foucart, ``Notes on compressed sensing,'' {\em available online}, 2009.

\bibitem{tao1}
E.~Cand\`{e}s and T.~Tao, ``Near-optimal signal recovery from random
  projections: Universal encoding strategies?,'' {\em IEEE Trans. Inf. Theory},
  vol.~52, no.~12, pp.~5406--5425, 2006.

\bibitem{gribonval1}
R.~Gribonval and M.~Nielsen, ``Sparse representations in unions of bases,''
  {\em IEEE Trans. Inf. Theory}, vol.~49, no.~12, pp.~3320--3325, 2003.

\bibitem{nsp}
R.~Gribonval and M.~Nielsen, {\em On the Strong Uniqueness of Highly Sparse
  Representations from Redundant Dictionaries}, vol.~3195 of {\em Lecture Notes
  in Computer Science}.
\newblock Springer, 2004.

\bibitem{boothby}
W.~M. Boothby, {\em An Introduction to Differentiable Manifolds and Riemannian
  Geometry}.
\newblock Elsevier, 2~ed., 2003.

\bibitem{milnor}
J.~Milnor, {\em Differential Topology}.
\newblock 1958 Lecture Notes.

\bibitem{mattila}
P.~Mattila, {\em Geometry of Sets and Measures in Euclidean Spaces}.
\newblock Cambridge University Press, 1995.

\bibitem{raskutti2010restricted}
G.~Raskutti, M.~J. Wainwright, and B.~Yu, ``Restricted eigenvalue properties
  for correlated {Gaussian} designs,'' {\em The Journal of Machine Learning
  Research}, vol.~11, pp.~2241--2259, 2010.

\bibitem{tao2}
E.~Cand\`{e}s, J.~Romberg, and T.~Tao, ``Robust uncertainty principles: Exact
  signal reconstruction from highly incomplete frequency information,'' {\em
  IEEE Trans. Inf. Theory}, vol.~52, no.~2, pp.~489--509, 2006.

\bibitem{gordon}
Y.~Gordon, {\em On Milman's inequality and random subspaces which escape
  through a mesh in $\mathbb{R}^n$}.
\newblock Springer, 1988.

\bibitem{rudelson}
M.~Rudelson and R.~Vershynin, ``On sparse reconstruction from {F}ourier and
  {G}aussian measurements,'' {\em Communications on Pure and Applied
  Mathematics}, vol.~61, no.~8, pp.~1025--1045, 2008.

\bibitem{dyson1962threefold}
F.~J. Dyson, ``The threefold way. algebraic structure of symmetry groups and
  ensembles in quantum mechanics,'' {\em Journal of Mathematical Physics},
  vol.~3, no.~6, pp.~1199--1215, 1962.

\bibitem{stojnic}
M.~Stojnic, ``Under-determined linear systems and $\ell_q$-optimization
  thresholds,'' {\em CoRR}, vol.~abs/1306.3774, 2013.

\bibitem{amelunxen2014living}
D.~Amelunxen, M.~Lotz, M.~B. McCoy, and J.~A. Tropp, ``Living on the edge:
  Phase transitions in convex programs with random data,'' {\em Inform.
  Inference}, 2014.

\bibitem{bickel2009simultaneous}
P.~J. Bickel, Y.~Ritov, and A.~B. Tsybakov, ``Simultaneous analysis of {Lasso}
  and {Dantzig} selector,'' {\em The Annals of Statistics}, pp.~1705--1732,
  2009.

\bibitem{vershynin2010introduction}
R.~Vershynin, ``Introduction to the non-asymptotic analysis of random
  matrices,'' {\em arXiv preprint arXiv:1011.3027}, 2010.

\bibitem{bai1993limit}
Z.~Bai and Y.~Yin, ``Limit of the smallest eigenvalue of a large dimensional
  sample covariance matrix,'' {\em The annals of Probability}, pp.~1275--1294,
  1993.

\bibitem{adamczak2011restricted}
R.~Adamczak, A.~E. Litvak, A.~Pajor, and N.~Tomczak-Jaegermann, ``Restricted
  isometry property of matrices with independent columns and neighborly
  polytopes by random sampling,'' {\em Constructive Approximation}, vol.~34,
  no.~1, pp.~61--88, 2011.

\bibitem{rauhut2010compressive}
H.~Rauhut, ``Compressive sensing and structured random matrices,'' {\em
  Theoretical foundations and numerical methods for sparse recovery}, vol.~9,
  pp.~1--92, 2010.

\bibitem{lqharmonic}
A.~Aldroubi, X.~Chen, and A.~M. Powell, ``Perturbations of measurement matrices
  and dictionaries in compressed sensing,'' {\em Applied and Computational
  Harmonic Analysis}, vol.~33, no.~2, pp.~282 -- 291, 2012.

\bibitem{stein_func}
E.~M. Stein and R.~Shakarchi, {\em Functional Analysis: Introduction to Further
  Topics in Analysis}.
\newblock Princeton University Press, 2011.

\bibitem{troppdis}
J.~A. Tropp, ``Topics in sparse approximation, {Ph.D}. dissertation,'' {\em
  Computational and Applied Mathematics, Univ. Texas at Austin}, August 2004.

\bibitem{candes2006stable}
E.~J. Candes, J.~K. Romberg, and T.~Tao, ``Stable signal recovery from
  incomplete and inaccurate measurements,'' {\em Communications on pure and
  applied mathematics}, vol.~59, no.~8, pp.~1207--1223, 2006.

\bibitem{candes2008restricted}
E.~J. Candes, ``The restricted isometry property and its implications for
  compressed sensing,'' {\em Comptes Rendus Mathematique}, vol.~346, no.~9,
  pp.~589--592, 2008.

\bibitem{blanchard2011compressed}
J.~D. Blanchard, C.~Cartis, and J.~Tanner, ``Compressed sensing: How sharp is
  the restricted isometry property?,'' {\em SIAM review}, vol.~53, no.~1,
  pp.~105--125, 2011.

\bibitem{raskutti2011minimax}
G.~Raskutti, M.~J. Wainwright, and B.~Yu, ``Minimax rates of estimation for
  high-dimensional linear regression over-balls,'' {\em Information Theory,
  IEEE Transactions on}, vol.~57, no.~10, pp.~6976--6994, 2011.

\bibitem{wu2012optimal}
Y.~Wu and S.~Verd{\'u}, ``Optimal phase transitions in compressed sensing,''
  {\em Information Theory, IEEE Transactions on}, vol.~58, no.~10,
  pp.~6241--6263, 2012.

\bibitem{donoho2009observed}
D.~Donoho and J.~Tanner, ``Observed universality of phase transitions in
  high-dimensional geometry, with implications for modern data analysis and
  signal processing,'' {\em Philosophical Transactions of the Royal Society A:
  Mathematical, Physical and Engineering Sciences}, vol.~367, no.~1906,
  pp.~4273--4293, 2009.

\bibitem{bayati2012universality}
M.~Bayati, M.~Lelarge, and A.~Montanari, ``Universality in polytope phase
  transitions and message passing algorithms,'' {\em arXiv preprint
  arXiv:1207.7321}, 2012.

\bibitem{optheory}
N.~I. Akhiezer, I.~M. Glazman, and M.~K. Nestell, {\em Theory of linear
  operators in Hilbert space}, vol.~1.
\newblock F. Ungar Publishing Company, 1961.

\bibitem{rudin}
W.~Rudin, {\em Real and Complex Analysis}.
\newblock McGraw-Hill, 3~ed., 1987.

\bibitem{munk}
J.~R. Munkres, {\em Topology}.
\newblock Upper Saddle River, NJ Prentice Hall, Inc. c2000, 2nd ed~ed.

\bibitem{stein}
E.~M. Stein and R.~Shakarchi, {\em Real Analysis: Measure Theory, Integration,
  and Hilbert Spaces}.
\newblock Princeton University Press, 2005.

\end{thebibliography}

\end{document}